\documentclass[12pt]{article}
\usepackage[english]{babel}
\usepackage{amsmath,amsthm,amssymb,amsfonts}
\usepackage{url,color}

\newtheorem{definition}{Definition}

\usepackage{cite}
\newtheorem{theorem}{Theorem}[section]
\newtheorem{remark}[theorem]{Remark}

\usepackage[colorlinks=true,linkcolor=blue,citecolor=red, urlcolor=green]{hyperref} 

\newcommand\ii{{\operatorname{i}}}

\begin{document}

\title{Direct approach to \\approximate conservation laws}

\author{M.~Gorgone, G.~Inferrera\\
\ \\
{\footnotesize Department of Mathematical and Computer Sciences,}\\
{\footnotesize Physical Sciences and Earth Sciences, University of Messina}\\
{\footnotesize Viale F. Stagno d'Alcontres 31, 98166 Messina, Italy}\\
{\footnotesize matteo.gorgone@unime.it; guinferrera@unime.it}
}

\date{Published in  \textit{Eur. Phys. J. Plus.}  \textbf{138}, 447 (2023).}

\maketitle

\begin{abstract}
In this paper, non--variational systems of differential equations containing 
small terms are considered, and a consistent approach for deriving approximate 
conservation laws through the introduction of approximate Lagrange multipliers 
is developed. The proposed formulation of the approximate direct method 
starts by assuming the Lagrange multipliers to be dependent on the small 
parameter; then, by expanding the dependent variables in power series of the 
small parameter, we consider the
consistent expansion of all the involved quantities (equations and Lagrange 
multipliers) in such a way the basic principles of perturbation analysis are 
not violated. Consequently, a theorem leading to the determination of 
approximate multipliers whence approximate conservation laws arise is 
proved, and the role of approximate Euler operators emphasized. Some 
applications of the procedure are presented.
\end{abstract}


\section{Introduction}\label{sec:intro}
The determination of conservation laws associated to mathematical models 
expressed in terms of differential equations
is important from many viewpoints, either analytical or numerical: 
investigation of integrability and linearization, introduction of potentials 
and nonlocally--related systems, analysis of solutions, accuracy of numerical 
methods 
\cite{Bressan,Dafermos,Lax,LeVeque,HairerLubichWanner,Iserlesetal,Ibragimov:CRC,BlumanCheviakovAnco}. 

A powerful and elegant tool for the derivation of conserved quantities is 
given by Lie group analysis of differential equations, and several strategies 
have been widely used 
\cite{Ibragimov:CRC,BlumanCheviakovAnco,Olver,Oliveri2010,GorgoneOliveriSpeciale}. 
In particular, given a system of differential 
equations arising from a variational principle, \emph{i.e.}, derived from a Lagrangian function, 
the Noether's theorem \cite{Noether} allows to obtain algorithmically conservation laws  by using an 
explicit formula involving the infinitesimals of the Lie point symmetries of the Lagrangian action and 
the Lagrangian itself.
The generators of the Lie point symmetries of the Lagrangian action leave 
invariant the extrema of the Lagrangian action and, hence, they are admitted 
as generators of  symmetries of the corresponding Euler--Lagrange equations; 
however, in general it is not guaranteed that Lie point symmetries admitted by 
the Euler--Lagrange equations are variational symmetries.
In \cite{Boyer}, an extension of Noether's theorem to construct conservation 
laws arising from invariance under generalized Lie symmetries 
\cite{BlumanCheviakovAnco,Olver} is provided. 

Despite their relevance, the application of Noether's theorem is limited to 
variational systems and is coordinate--dependent. Nevertheless, there are many evolution equations,  
such as the classical heat equation, the Burgers equation, etc. \cite{AndersonDuchamp}, not admitting a 
Lagrangian function. Thus, for such cases one needs alternative methods. 

In \cite{Ibragimov2007}, in order to investigate non--variational problems, 
the concepts of \emph{formal} Lagrangian and \emph{self--adjointness} have been 
introduced, and an explicit formula allowing to construct conservation laws 
associated to the symmetries of any nonlinearly self--adjoint system of 
equations has been given. In detail, the procedure requires to consider a 
self--adjoint system of a set of (linear and/or nonlinear) differential 
equations and construct a formal Lagrangian for the  starting system together 
with its self--adjoint system. It is proved that the adjoint system inherits 
all the Lie point symmetries of the original system. Therefore, this 
method is within the framework of Lie group analysis too.

In \cite{AncoBluman2002a,AncoBluman2002b}, a direct and algorithmic method 
which does not involve an action principle, and produces conservation laws by 
using the same independent and dependent variables as the original equations, 
has been proposed. This approach consists in solving linear determining 
equations whose unknowns are the so--called multipliers, and each solution of 
these equations leads to a conservation law.

The symmetry--based approaches for the derivation of conservation laws for differential 
equations need some adaptations when dealing with 
differential equations involving terms of different orders of magnitude. 
In fact, the occurrence of small terms has the effect of destroying some 
important symmetries. In the last decades, some approximate symmetry theories 
have been proposed \cite{BGI-1988,FS-1989,DSGO-2018}. In particular, in 
\cite{DSGO-2018}, a new approximate symmetry approach, which is consistent 
with the basic tenets of perturbation analysis \cite{Nayfeh}, and such that the relevant 
properties of exact Lie symmetries of differential
equations are inherited, has been introduced. This method requires to expand 
in power series of the small parameter both the dependent variables, as done in 
classical perturbation analysis, and the Lie generator, so that a new 
approximate invariance with respect to an approximate Lie generator can be 
defined. Some applications of the new approach to approximate Lie symmetries 
of differential equations can be found in 
\cite{Gorgone-2018,GorgoneOliveriEDJE-2018,GorgoneOliveriZAMP-2021,GO-Maths-2021}. 

In \cite{GO-Maths-2021}, as far as the problem of determining approximate 
conservations laws by means of Lie group methods is concerned, the consistent 
approach \cite{DSGO-2018} has been applied to perturbed variational problems, 
stating a new approximate Noether's theorem leading to the construction of 
approximate conserved quantities.

The aim of this paper is to consider differential equations involving small 
terms that do not arise from a variational principle, and consistently extend  
the procedure developed in \cite{AncoBluman2002a,AncoBluman2002b} to an 
approximate context, in order to find approximate conservation laws by a 
direct computation. 

The problem of determining approximate conservation laws of non--varia\-tional 
problems involving a small parameter has been faced in \cite{Jamal-2019}, 
where an expansion of the Lagrange multipliers is introduced leading to 
approximate conservation laws; basically, the approach therein moves within 
the framework of Baikov--Gazizov--Ibragimov method for approximate symmetries 
\cite{BGI-1988}. The Baikov--Gazizov--Ibragimov approach to approximate symmetries is also the framework where Tarayrah \cite{Tara} moves in deriving approximate conservation laws.
In \cite{Jamal-2019}, there have been also considered some examples 
where the dependent variables are expanded too, as the method by Fushchich--Shtelen \cite{FS-1989} does. 

The approach to approximate conservation laws here presented is somehow different and moves 
in the framework already exploited in 
\cite{DSGO-2018,GO-Maths-2021}, where the symmetries of the differential equations involving small terms are
faced in such a way the basic principles of perturbative analysis are preserved. In fact, we expand the dependent variables, and, assuming the 
Lagrange multipliers to be dependent on the small parameter, we accordingly 
expand them. Therefore, the results are somehow different from those that can be recovered 
using the methods described in \cite{Jamal-2019,Tara}.
Remarkably, the method here developed requires the introduction of an 
approximate Euler operator (see \cite{GO-Maths-2021}, where an approximate 
Noether theorem has been stated).

The plan of the paper is as follows. In Section~\ref{sec:2}, for the reader's convenience, a brief sketch of the approximate framework about conservation laws 
and the direct method is given. 
In Section~\ref{sec:3}, after  briefly reviewing the other direct approaches to approximate conservation laws, we introduce the approximate multipliers accordingly to a consistent approach, and a new approximate direct  
procedure is developed. In Section~\ref{sec:4}, a perturbed nonlinear diffusion equation is considered and the admitted approximate conservation laws are computed by means of the various methods with the aim of highlighting the differences. In Section~\ref{sec:5}, some illustrative applications 
of physical interest are presented, and approximate multipliers with the 
corresponding approximate conservation laws are determined. Finally, 
Section~\ref{sec:6} contains our conclusions.

\section{Theoretical framework and notation}\label{sec:2}

In this Section, to keep the paper self--contained and fix the notation, we briefly recall some basic facts 
regarding conservation laws within the approximate framework. 

Let
\begin{equation}
\label{system}
\boldsymbol\Delta\left(\mathbf{x},\mathbf{u}^{(r)};\varepsilon\right)=\mathbf{0}
\end{equation}
be a system of differential equations of order $r$ where some terms include a parameter $\varepsilon\ll 1$;
functions $\boldsymbol\Delta\equiv(\Delta^1,\ldots,\Delta^q)$, assumed to be sufficiently smooth, depend on  the independent 
variables  $\mathbf{x}\equiv(x_1,\ldots,x_n)\in X\subseteq\mathbb{R}^n$, the dependent ones 
$\mathbf{u}^{(0)}=\mathbf{u}\equiv(u_1,\ldots,u_m)\in U\subseteq\mathbb{R}^m$, and the derivatives $\mathbf{u}^{(r)}=\left\{\frac{\partial^{|J|} u_{\alpha}}{\partial x^{j_1}_1\ldots\partial x^{j_n}_n}:\;\alpha=1,\ldots,m,\;|J|=0,\ldots,r\right\}\allowbreak\in U^{(r)}\subseteq\mathbb{R}^N$, with $N=m\binom{n+r}{r}$,  of the latter with respect to the former up to the order $r$, where $J\equiv(j_1,\ldots,j_n)$ is a multi--index and $|J|=j_1+\dots+j_n$.

Differential equations like (\ref{system}) are often investigated by means of 
 perturbative techniques~\cite{Nayfeh}, whence we need to expand the variables $\mathbf{u}^{(r)}$ in power series of $\varepsilon$ up to some finite order $p$, \emph{i.e.},
\begin{equation}
\label{expansion_u_der}
\mathbf{u}^{(r)}(\mathbf{x};\varepsilon)=\sum_{k=0}^p\varepsilon^k \mathbf{u}^{(r)}_{(k)}(\mathbf{x})
+O(\varepsilon^{p+1}),
\end{equation}
with $\mathbf{u}^{(r)}_{(k)}\equiv(u^{(r)}_{(k)1},\ldots,u^{(r)}_{(k)N})$.

Let us recall the notion of \emph{approximate function} \cite{DSGO-2018}.
\begin{definition}[Approximate function]
\label{approx-func}
Let $f(\mathbf{x},\mathbf{u}^{(r)};\varepsilon)$ be a $C^\infty$ function depending on the independent variables $\mathbf{x}\in X\subseteq\mathbb{R}^n$, and the derivatives $\mathbf{u}^{(r)}\in U^{(r)}\subseteq\mathbb{R}^N$ ($\mathbf{u}^{(0)}=\mathbf{u}\in U\subseteq\mathbb{R}^m$), and the small parameter $\varepsilon\in\mathbb{R}$; in the following, we will consider such kind of functions locally in a neighborhood of $\varepsilon=0$.

Using  \eqref{expansion_u_der}, and expanding the function $f(\mathbf{x},\mathbf{u}^{(r)};\varepsilon)$ in power series of $\varepsilon$, we obtain 
\begin{equation}
\begin{aligned}
f(\mathbf{x},\mathbf{u}^{(r)};\varepsilon)&=\sum_{k=0}^p\sum_{|\sigma|=k}
\frac{\varepsilon^{\sigma_0}}{\sigma_0!}\left(\prod_{i=1}^N \frac{(u^{(r)}_i-u^{(r)}_{(0)i})^{\sigma_i}}{\sigma_i!}\right)
\left.\frac{\partial^{|\sigma|}f(\mathbf{x},\mathbf{u}^{(r)};\varepsilon)}{\partial \varepsilon^{\sigma_0}\partial u^{(r)\sigma_1}_1\cdots\partial u^{(r)\sigma_N}_N}\right|_{\varepsilon=0}\\
&+O(\varepsilon^{p+1}),
\end{aligned}
\end{equation}
$\sigma$ being the multi--index $(\sigma_0,\sigma_1,\ldots,\sigma_N)$, and $|\sigma|=\sigma_0+\sigma_1+\dots+
\sigma_N$. 

By means of the positions
\begin{equation}
\begin{aligned}
&{f}_{(0)}(\mathbf{x},\mathbf{u}^{(r)}_{(0)})=\left.f(\mathbf{x},\mathbf{u}^{(r)};\varepsilon)\right|_{\varepsilon=0},\\
&{f}_{(k)}(\mathbf{x},\mathbf{u}^{(r)}_{(0)})=\left.\frac{\partial^k f(\mathbf{x},\mathbf{u}^{(r)};\varepsilon)}{\partial \varepsilon^k}\right|_{\varepsilon=0},
\end{aligned}
\end{equation}
since it is
\begin{equation}
\left.\frac{\partial^{|\sigma|}f(\mathbf{x},\mathbf{u}^{(r)};\varepsilon)}{\partial \varepsilon^{\sigma_0}\partial u_1^{(r)\sigma_1}\cdots\partial 
u_N^{(r)\sigma_N}}\right|_{\varepsilon=0}=\frac{\partial^{|\sigma|-\sigma_0} {f}_{(\sigma_0)}(\mathbf{x},\mathbf{u}^{(r)}_{(0)})}{\partial 
u_{(0)1}^{(r)\sigma_1}\cdots\partial u_{(0)N}^{(r)\sigma_N}},
\end{equation}
we have
\begin{equation}
\label{exp-1}
\begin{aligned}
f(\mathbf{x},\mathbf{u}^{(r)};\varepsilon)&=\sum_{k=0}^p\sum_{|\sigma|=k}
\frac{\varepsilon^{\sigma_0}}{\sigma_0!}\left(\prod_{i=1}^N \frac{(u_i^{(r)}-u_{(0)i}^{(r)})^{\sigma_i}}{\sigma_i!}\right)
\frac{\partial^{|\sigma|-\sigma_0}{f}_{(\sigma_0)}(\mathbf{x},\mathbf{u}^{(r)}_{(0)})}{\partial u_{(0)1}^{(r)\sigma_1}\cdots\partial u_{(0)N}
^{(r)\sigma_N}}\\
&+O(\varepsilon^{p+1});
\end{aligned}
\end{equation}
therefore, the expansion in power series of $\varepsilon$ is characterized (up to the order $p$ in $\varepsilon$) by 
$p+1$ functions depending on $(\mathbf{x},\mathbf{u}^{(r)}_{(0)})$. Such an expansion can be written as
\begin{equation}
\label{exp-2}
\begin{aligned}
&f(\mathbf{x},\mathbf{u}^{(r)};\varepsilon)=\sum_{k=0}^p \varepsilon^k \widetilde{f}_{(k)}\left(\mathbf{x},\mathbf{u}^{(r)}_{(0)},\ldots, \mathbf{u}^{(r)}_{(k)}\right)
+O(\varepsilon^{p+1}),
\end{aligned}
\end{equation}
where $\widetilde{f}_{(k)}$ $(k>0)$ are suitable polynomials in $\mathbf{u}^{(r)}_{(1)},\ldots,\mathbf{u}^{(r)}_{(k)}$ with 
coefficients given by ${f}_{(0)}(\mathbf{x},\mathbf{u}^{(r)}_{(0)})$, \ldots, ${f}_{(k)}(\mathbf{x},\mathbf{u}^{(r)}_{(0)})$ and their derivatives with respect to 
$\mathbf{u}^{(r)}_{(0)}$. More precisely, the functions $\widetilde{f}_{(k)}$  are defined as follows:
\begin{equation}
\begin{aligned}
&\widetilde{f}_{(0)}={f}_{(0)},\\
&\widetilde{f}_{(k+1)}=\frac{1}{k+1}\mathcal{R}[\widetilde{f}_{(k)}],
\end{aligned}
\end{equation} 
$\mathcal{R}$ being a linear recursion operator, satisfying the \emph{Leibniz rule}, defined as
\begin{equation}
\label{R_operator}
\begin{aligned}
&\mathcal{R}\left[\frac{\partial^{|\tau|}{f}_{(k)}(\mathbf{x},\mathbf{u}^{(r)}_{(0)})}{\partial u_{(0)1}^{(r)\tau_1}\dots\partial u_{(0)N}^{(r)\tau_N}}
\right]=\frac{\partial^{|\tau|}{f}_{(k+1)}(\mathbf{x},\mathbf{u}^{(r)}_{(0)})}{\partial u_{(0)1}^{(r)\tau_1}\dots\partial u_{(0)N}^{(r)\tau_N}}\\
&\phantom{\mathcal{R}\left[\frac{\partial^{|\tau|}{f}_{(k)}(\mathbf{x},\mathbf{u}^{(r)}_{(0)})}{\partial u_{(0)1}^{(r)\tau_1}\dots\partial u_{(0)N}^{(r)\tau_N}}
\right]}
+\sum_{i=1}^N\frac{\partial}{\partial u^{(r)}_{(0)i}}\left(\frac{\partial^{|\tau|} {f}_{(k)}(\mathbf{x},\mathbf{u}^{(r)}_{(0)})}{\partial u_{(0)1}
^{(r)\tau_1}\dots\partial u_{(0)N}^{(r)\tau_N}}\right)u^{(r)}_{(1)i},\\
&\mathcal{R}[u^{(r)}_{(k)j}]=(k+1)u^{(r)}_{(k+1)j},
\end{aligned}
\end{equation}
where $k\ge 0$,  $j=1,\ldots, N$, $|\tau|=\tau_1+\cdots+\tau_N$.
\end{definition}

\begin{definition}
Given any couple of smooth functions $f$ and $g$, the notation 
$f\approx g$ means that $f$ and $g$ have the same Taylor expansion up to the order $p$ in $\varepsilon$.
\end{definition}
Thence, by means of Definition \ref{approx-func}, we can define the notion of conservation law in the approximate framework.

\begin{definition}[Approximate conservation law]
\label{approx-claw}
Given a system of differential equations,
\begin{equation}
\label{sysgen}
\boldsymbol{\Delta}\left(\mathbf{x},\mathbf{u}^{(r)};\varepsilon\right)\approx \sum_{k=0}^p\varepsilon^k\widetilde{\boldsymbol\Delta}_{(k)}\left(\mathbf{x},
\mathbf{u}^{(r)}_{(0)},
\ldots,\mathbf{u}^{(r)}_{(k)}\right)=\mathbf{0},
\end{equation}
an \emph{approximate conservation law} of order $r$ compatible with the system \eqref{sysgen} is an approximate divergence expression
\begin{equation}
\label{conservationlaw} 
 \sum_{k=0}^p\varepsilon^k\left(\sum_{i=1}^n D_{i}\left(\widetilde{\Phi}^i_{(k)}\left(\mathbf{x},
\mathbf{u}^{(r-1)}_{(0)},
\ldots,\mathbf{u}^{(r-1)}_{(k)}\right)\right)\right) =O(\varepsilon^{p+1}),
\end{equation}
holding for all solutions of system \eqref{sysgen}, where 
\[
\sum_{k=0}^p\varepsilon^k\widetilde{\boldsymbol\Delta}_{(k)}\left(\mathbf{x},
\mathbf{u}^{(r)}_{(0)},
\ldots,\mathbf{u}^{(r)}_{(k)}\right)
\]
and
\[
\sum_{k=0}^p \varepsilon^k  \widetilde{\Phi}^i_{(k)}\left(\mathbf{x},
\mathbf{u}^{(r-1)}_{(0)},
\ldots,\mathbf{u}^{(r-1)}_{(k)}\right), \qquad i=1,\ldots,n
\]
are the expansions at order $p$ of the differential equations $\boldsymbol{\Delta}\left(\mathbf{x},\mathbf{u}^{(r)};\varepsilon\right)$ and fluxes $\Phi^i\left(\mathbf{x},
\mathbf{u}^{(r-1)};\varepsilon\right)$ of the approximate conservation law, respectively, according to \eqref{exp-1}--\eqref{exp-2}, along with the approximate Lie derivative \cite{DSGO-2018} defined as
\begin{equation}
\label{Lie-derivative}
D_i=\frac{{D}}{{D} x_i}=\frac{\partial}{\partial x_i}+\sum_{k=0}^p\sum_{\alpha=1}^m \left(u_{(k)\alpha,i}
\frac{\partial}{\partial u_{(k)\alpha}}+\sum_{j=1}^n u_{(k)\alpha,ij} \frac{\partial}{\partial u_{(k)\alpha,j}}+\dots\right),
\end{equation}
with $\displaystyle u_{(k)\alpha,i}=\frac{\partial u_{(k)\alpha}}{\partial x_{i}}$, 
$\displaystyle u_{(k)\alpha,ij}=\frac{\partial^2 u_{(k)\alpha}}{\partial x_{i}\partial x_j}$, \ldots
\end{definition}
Some properties about exact conservation laws can be inherited in the approximate context:
\begin{itemize}
\item an approximate conservation law is \emph{trivial} if each of its approximate fluxes is $O(\varepsilon^{p+1})$ on the solutions of the given system of differential equations, or the approximate conservation law is $O(\varepsilon^{p+1})$ identically;
\item two approximate conservation laws are said \emph{equivalent} if their linear combination is a trivial approximate conservation law.
\end{itemize}
The notion of trivial approximate conservation laws allows us for the introduction of \emph{linearly dependent} approximate conservation laws.
\begin{definition}
A finite set $\mathcal{S}$ of approximate conservation laws is linearly dependent if there exists a set of constants, not all zero, such that the linear combination of the elements in $\mathcal{S}$ is trivial. In this case, at least one of the approximate conservation laws in $\mathcal{S}$ can be expressed as a linear combination of the remaining ones. 
\end{definition}
For variational problems containing small terms, \emph{i.e.}, differential equations derived from a perturbed Lagrangian function, the determination of approximate conservation laws can be performed by the approximate Noether's theorem \cite{GO-Maths-2021}, which establishes a correspondence between the approximate Lie symmetries of the perturbed action integral and approximate conservation laws through an explicit formula. 

As said before, the most relevant disadvantage of Noether's procedure for the determination of conservation laws relies on the fact that it applies only to differential equations admitting a Lagrangian formulation.

In general, for non--variational differential equations, conservation laws may be found by means of a direct approach \cite{AncoBluman2002a,AncoBluman2002b}, which makes use of the so--called \emph{multipliers}. 
By reformulating this method in the approximate framework, we have to consider a system like \eqref{system} of differential equations, and look for sets of non--singular (when evaluated on the solutions of the system) multipliers $\Lambda^\nu\left(\mathbf{x},\mathbf{u}^{(r)};\varepsilon\right)$ ($\nu=1,\ldots,q$), that are functions of the independent and dependent variables, as well as derivatives (up to some finite order) that we 
explicitly assume to depend also on $\varepsilon$  such that the linear combination
\begin{equation}
\sum_{\nu=1}^q\left(\Lambda^\nu\left(\mathbf{x},\mathbf{u}^{(r)};\varepsilon\right) {\Delta}^\nu\left(\mathbf{x},\mathbf{u}^{(r)};\varepsilon\right)\right)\equiv\sum_{i=1}^nD_{i}\left({\Phi}^i\left(\mathbf{x},
\mathbf{u}^{(r-1)};\varepsilon\right)\right)\approx 0
\end{equation}
is an approximate divergence expression. 
The crucial requirement of this procedure is that this linear combination must be annihilated by the Euler operators associated to all dependent variables, where the dependent variables and their derivatives are replaced by arbitrary functions depending on the independent variables. Then, all the sets of multipliers linked to conservation laws can be found algorithmically by solving a linear system of determining equations. It is worth of being underlined that the consistent approach we use in this paper requires the introduction of approximate Euler operators \cite{GO-Maths-2021}, as will be shown in Section~\ref{sec:3}.

\section{Direct approaches to approximate conservation laws}
\label{sec:3}
In this Section, we present our consistent direct approach for deriving approximate conservation laws. In order to clarify the differences with the methods in the literature and used by many authors, we shall start by
briefly reviewing the methods which definitely relie on the approach to approximate symmetries proposed by Baikov--Gazizov--Ibragimov \cite{BGI-1988}, where the dependent variables are not expanded, or the one by Fushchich--Shtelen \cite{FS-1989}, where the expansion of the dependent variables is performed. Hereafter, we denote as Approach A the direct method for determining approximate conservation laws within the Baikov--Gazizov--Ibragimov \cite{BGI-1988} framework for approximate symmetries, whereas with Approach B the direct method within the Fushchich--Shtelen \cite{FS-1989} one.

In the Approach A where the expansion of the dependent variables is not considered, and followed in \cite{Jamal-2019,Tara}, given a system of differential equations involving a small parameter 
$\varepsilon$, say 
\begin{equation}
\label{sys-BGI}
\Delta^\nu\left(\mathbf{x},\mathbf{u}^{(r)};\varepsilon\right)\equiv \sum_{k=0}^p\varepsilon^k\Delta^\nu_{(k)}\left(\mathbf{x},\mathbf{u}^{(r)}
\right)=O(\varepsilon^{p+1}),\qquad \nu=1,\ldots,q,
\end{equation}
the expansion of the \emph{Lagrange multipliers} has the form
\begin{equation}
\label{approx-mult-BGI}
\begin{aligned}
\Lambda^\nu\left(\mathbf{x},\mathbf{u}^{(r)};\varepsilon\right)\equiv \sum_{k=0}^p\varepsilon^k\Lambda^\nu_{(k)}\left(\mathbf{x},\mathbf{u}^{(r)}\right),\qquad \nu=1,\ldots,q,
\end{aligned}
\end{equation}
whereas, the \emph{Euler operators} are
\begin{equation}
\label{euler-BGI}
\begin{aligned}
E_{u_{\alpha}}&=\frac{\partial }{\partial u_{\alpha}}
-\sum_{i=1}^nD_{i}\left(\frac{\partial 
}{\partial u_{\alpha,i}}\right)\\
&+\ldots+(-1)^r\sum_{i_1=1}^n\ldots\sum_{i_r=i_{r-1}}^nD_{i_1}\dots D_{i_r}\left(\frac{\partial 
}{\partial u_{\alpha,i_1\dots i_r}}\right),
\end{aligned}
\end{equation}
with $\alpha=1,\ldots,m$. By applying the algorithmic direct procedure, the recovered approximate multipliers (\ref{approx-mult-BGI}) yield an approximate divergence expression for (\ref{sys-BGI})
\begin{equation}
\begin{aligned}
\sum_{k=0}^p\varepsilon^k&\left(\sum_{\ell=0}^k\sum_{\nu=1}^q\Lambda^\nu_{(\ell)}\left(\mathbf{x},\mathbf{u}^{(r)}\right)\Delta^\nu_{(k-\ell)}\left(\mathbf{x},\mathbf{u}^{(r)}
\right)\right)\equiv\\
&\equiv \sum_{k=0}^p\varepsilon^k\sum_{i=1}^nD_{i}\left({\Phi}^i_{(k)}\left(\mathbf{x},
\mathbf{u}^{(r-1)}\right)\right)=O(\varepsilon^{p+1}).
\end{aligned}
\end{equation}
We observe that the expanded Lagrange multipliers are not fully consistent with the principles of perturbation analysis since the dependent variables are not expanded in power series of the small parameter;
as proved in \cite{DSGO-2018} (Remark 4) discussing approximate Lie symmetries, this may lead either to inconsistent or less general results since terms with different orders of magnitude are put together and dealt as terms with the same order in $\varepsilon$.

On the contrary, following the Approach B, if the expansion (\ref{expansion_u_der}) of the dependent variables and their derivatives is inserted in system (\ref{system}), thus obtaining  
\begin{equation}
\Delta^\nu\left(\mathbf{x},\mathbf{u}^{(r)};\varepsilon\right)\approx \sum_{k=0}^p\varepsilon^k\Delta^\nu_{(k)}\left(\mathbf{x},
\mathbf{u}^{(r)}_{(0)},
\ldots,\mathbf{u}^{(r)}_{(k)}\right)=O(\varepsilon^{p+1})
\end{equation}
($\nu=1,\ldots,q,$), and terms at each order of approximation are separated, we obtain a coupled system to be solved in a hierarchy:
\begin{equation}
\label{coupled-sys}
\Delta^\nu_{(k)}\left(\mathbf{x},\mathbf{u}^{(r)}_{(0)},
\ldots,\mathbf{u}^{(r)}_{(k)}\right)=0, \qquad k=0,\ldots,p,\qquad \nu=1,\ldots,q.
\end{equation}
Accordingly to this procedure, 
the \emph{approximate multipliers} of system (\ref{system}) are defined as the \emph{exact multipliers}
\begin{equation}
\label{approx-mult-FS}
\Lambda^\nu_{(k)}\equiv \Lambda^\nu_{(k)}\left(\mathbf{x},\mathbf{u}^{(r)}_{(0)},
\ldots,\mathbf{u}^{(r)}_{(p)}\right), \qquad k=0,\ldots,p,\qquad \nu=1,\ldots,q
\end{equation}
of system (\ref{coupled-sys}) \cite{Jamal-2019}.
In this case, the \emph{Euler operators} have the form:
\begin{equation}
\label{euler-FS}
\begin{aligned}
E_{u_{(k)\alpha}}&=\frac{\partial }{\partial u_{(k)\alpha}}
-\sum_{i=1}^nD_{i}\left(\frac{\partial 
}{\partial u_{(k)\alpha,i}}\right)\\
&+\ldots+(-1)^r\sum_{i_1=1}^n\ldots\sum_{i_r=i_{r-1}}^nD_{i_1}\dots D_{i_r}\left(\frac{\partial 
}{\partial u_{(k)\alpha,i_1\dots i_r}}\right),
\end{aligned}
\end{equation}
with $k=0,\ldots,p\,$ and $\,\alpha=1,\ldots,m$.

Thence, the recovered approximate multipliers (\ref{approx-mult-FS}) 
yield an approximate divergence expression for (\ref{system})
\begin{equation}
\begin{aligned}
\sum_{\nu=1}^q\sum_{k=0}^p&\Lambda^\nu_{(k)}\left(\mathbf{x},\mathbf{u}^{(r)}_{(0)},
\ldots,\mathbf{u}^{(r)}_{(p)}\right)\Delta^\nu_{(k)}\left(\mathbf{x},\mathbf{u}^{(r)}_{(0)},
\ldots,\mathbf{u}^{(r)}_{(k)}\right)\equiv\\
&\equiv\sum_{i=1}^nD_{i}\left({\Phi}^i\left(\mathbf{x},\mathbf{u}^{(r-1)}_{(0)},
\ldots,\mathbf{u}^{(r-1)}_{(p)}\right)\right)=0,
\end{aligned}
\end{equation}
where the \emph{approximate} fluxes ${\Phi}^i\left(\mathbf{x},\mathbf{u}^{(r-1)}_{(0)},
\ldots,\mathbf{u}^{(r-1)}_{(p)}\right)$ are the exact fluxes corresponding to the exact multipliers of the approximate system (\ref{coupled-sys}).
Approach B is of course consistent with the principles of perturbation analysis. However, besides requiring 
a lot of algebra (it needs the use of $m\cdot(p+1)$ Euler operators!), the basic assumption that system (\ref{coupled-sys}) is fully coupled is too strong, since the equations at a level
are not influenced by those at higher levels. In addition, there is no possibility to work in a hierarchy: for
instance, if one computes first order approximate conservation laws, and then searches for second order 
approximate conservation laws, all the work must be done from the very beginning.

Our aim is to introduce a method that, besides being coherent with perturbation analysis, does not require a huge computational cost. Essentially, we combine the direct procedure \cite{AncoBluman2002a,AncoBluman2002b} with the consistent approach to approximate Lie symmetries introduced in \cite{DSGO-2018}.


According to Definition~\ref{approx-func}, let us start introducing the \emph{approximate multipliers}.

\begin{definition}[Approximate multiplier]
Let $\{\Lambda^\nu(\mathbf{x},\mathbf{u}^{(s)};\varepsilon): \nu=1,\ldots,q\}$ be a set of $C^\infty$ functions depending on a small parameter $\varepsilon$.

Functions $\Lambda^\nu(\mathbf{x},\mathbf{u}^{(s)};\varepsilon)$ $(\nu=1,\ldots,q)$ are called \emph{approximate multipliers} depending on $s$--th order derivatives associated to the system \eqref{sysgen} if, after expanding in pertubation series of $\varepsilon$ up to the order $p$, according to \eqref{exp-1}--\eqref{exp-2}, \emph{i.e.},
\begin{equation}
\label{approx-mult-us}
{\Lambda}^\nu\left(\mathbf{x},
\mathbf{u}^{(s)};\varepsilon\right)=\sum_{k=0}^p\varepsilon^k\widetilde{\Lambda}^\nu_{(k)}\left(\mathbf{x},
\mathbf{u}^{(s)}_{(0)},
\ldots,\mathbf{u}^{(s)}_{(k)}\right)+O(\varepsilon^{p+1}),\quad \nu=1,\ldots,q,
\end{equation}
the relation 
\begin{equation}
\sum_{k=0}^p\varepsilon^k\left(\sum_{\ell=0}^k\sum_{\nu=1}^q\left(\widetilde{\Lambda}^\nu_{(\ell)} \widetilde{{\Delta}}_{(k-\ell)}^\nu\right)- \sum_{j=1}^n D_j \widetilde{\Phi}^j_{(k)}\right)=O(\varepsilon^{p+1})
\end{equation} 
holds for arbitrary $\mathbf{u}^{(s)}_{(\ell)}(\mathbf{x})$ and some suitable functions $\widetilde{\Phi}^j_{(k)}\left(\mathbf{x},\mathbf{u}_{(0)}^{(s-1)},\ldots,\mathbf{u}_{(k)}^{(s-1)}\right)$. 

Then, if ${\Lambda}^\nu\left(\mathbf{x},
\mathbf{u}^{(s)};\varepsilon\right)$ are non--singular, an approximate conservation law can be recovered:
\begin{equation}
\label{divergence_form}
\sum_{k=0}^p\varepsilon^k\left(\sum_{\ell=0}^k\sum_{\nu=1}^q\left(\widetilde{\Lambda}^\nu_{(\ell)} \widetilde{{\Delta}}_{(k-\ell)}^\nu\right)\right)\equiv \sum_{k=0}^p\varepsilon^k\sum_{j=1}^n D_j \widetilde{\Phi}^j_{(k)}= O(\varepsilon^{p+1}).
\end{equation} 
\end{definition}
\begin{definition}[Trivial approximate multiplier]
An approximate multiplier $\allowbreak{\Lambda}^\nu\left(\mathbf{x},
\mathbf{u}^{(s)};\varepsilon\right)\approx\sum_{k=0}^p\varepsilon^k\widetilde{\Lambda}^\nu_{(k)}\left(\mathbf{x},
\mathbf{u}^{(s)}_{(0)},
\ldots,\mathbf{u}^{(s)}_{(k)}\right)$ of system \eqref{sysgen} is a \emph{trivial} approximate multiplier if $\widetilde{\Lambda}^\nu_{(0)}\left(\mathbf{x},\mathbf{u}^{(s)}_{(0)}\right)=0$.
\end{definition}
\begin{remark}
If ${\Lambda}^\nu\left(\mathbf{x},
\mathbf{u}^{(s)};\varepsilon\right)\approx\sum_{k=0}^p\varepsilon^k\widetilde{\Lambda}^\nu_{(k)}\left(\mathbf{x},
\mathbf{u}^{(s)}_{(0)},
\ldots,\mathbf{u}^{(s)}_{(k)}\right)$ is a non--trivial approximate multiplier of system \eqref{sysgen}, then $\widetilde{\Lambda}^\nu_{(0)}\left(\mathbf{x},\mathbf{u}^{(s)}_{(0)}\right)$ is an exact multiplier of the unperturbed system ${\boldsymbol\Delta}\left(\mathbf{x},
\mathbf{u}^{(s)}_{(0)};0\right)=\mathbf{0}$. In such a case, we say that the exact multiplier $\widetilde{\Lambda}^\nu_{(0)}\left(\mathbf{x},\mathbf{u}^{(s)}_{(0)}\right)$ is stable with respect to the perturbation considered.
\end{remark}
\begin{remark}
If ${\Lambda}^\nu\left(\mathbf{x},\mathbf{u}^{(s)};\varepsilon\right)$ is an approximate multiplier of system \eqref{sysgen}, then $\varepsilon{\Lambda}^\nu\left(\mathbf{x},
\mathbf{u}^{(s)};\varepsilon\right)$ is an approximate multiplier too, but the converse is not in general true.
\end{remark}
Since we are interested in determining approximate multipliers yielding non--trivial approximate conservation laws, it is necessary that the system \eqref{sysgen} is solved with respect to some leading derivatives, \emph{i.e.}, it is in Cauchy--Kovalevskaya
form \cite{BlumanCheviakovAnco}.

As it is well known for unperturbed differential equations, the direct approach requires to use the Euler operators because they annihilate any divergence expression \cite{BlumanCheviakovAnco}. Therefore, in the approximate framework, we need the conditions allowing us to recover the approximate multipliers, and a new concept of approximate Euler operator \cite{GO-Maths-2021} has to be defined. Then, we can state the following theorem.
\begin{theorem}
\label{teo1}
A set of non--singular approximate multipliers $\{\Lambda^\nu(\mathbf{x},\mathbf{u}^{(r)};\varepsilon): \nu=1,\ldots,q\}$ yields an approximate conservation law for the system \eqref{sysgen} if and only if the set of relations
\begin{equation}
\label{euler-cond}
E_{u_{(0)\alpha}}\left(\sum_{k=0}^p\varepsilon^k\left(\sum_{\ell=0}^k\sum_{\nu=1}^q\left(\widetilde{\Lambda}^\nu_{(\ell)} \widetilde{{\Delta}}_{(k-\ell)}^\nu\right)\right)\right)\approx 0,\qquad \alpha=1,\ldots,m
\end{equation}
holds for arbitrary $\mathbf{u}^{(r)}_{(k)}(\mathbf{x})$ $(k=0,\ldots,p)$, where 
\begin{equation}
\label{euler-approx}
\begin{aligned}
E_{u_{(0)\alpha}}&=\frac{\partial }{\partial u_{(0)\alpha}}
-\sum_{i=1}^nD_{i}\left(\frac{\partial 
}{\partial u_{(0)\alpha,i}}\right)\\
&+\ldots+(-1)^r\sum_{i_1=1}^n\ldots\sum_{i_r=i_{r-1}}^nD_{i_1}\dots D_{i_r}\left(\frac{\partial 
}{\partial u_{(0)\alpha,i_1\dots i_r}}\right)
\end{aligned}
\end{equation}
are the \emph{approximate Euler operators}.

\end{theorem}
\begin{proof}
Theorem \ref{teo1} just extends the application of the direct method for unperturbed differential equations to the approximate context.
We only have to prove how the approximate Euler operators \eqref{euler-approx} are derived.

In order to achieve our purpose, let $\mathcal{L}\left(\mathbf{x},\mathbf{u}^{(r)};\varepsilon\right)$ be a perturbed $r$-th order Lagrangian function that, expanding in power series of $\varepsilon$, according to \eqref{exp-1}--\eqref{exp-2}, reads
\begin{equation}
\label{lag-expansion}
\mathcal{L}\left(\mathbf{x},\mathbf{u}^{(r)};\varepsilon\right)
\approx \sum_{k=0}^p\varepsilon^k \widetilde{\mathcal{L}}_{(k)}\left(\mathbf{x},\mathbf{u}_{(0)}^{(r)},\ldots,\mathbf{u}_{(k)}^{(r)}\right).
\end{equation}
Then, let us consider the corresponding perturbed Lagrangian action
\begin{equation}
\label{action}
\begin{aligned}
\mathcal{J}\left(\mathbf{x},\mathbf{u}^{(r)};\varepsilon\right)&= \int_\Omega
\mathcal{L}\left(\mathbf{x},\mathbf{u}^{(r)};\varepsilon\right)d\mathbf{x} \approx\\
&\approx\int_\Omega \left(
\sum_{k=0}^p\varepsilon^k\widetilde{\mathcal{L}}_{(k)}\left(\mathbf{x},\mathbf{u}^{(r)}_{(0)},\ldots,\mathbf{u}^{(r)}_{(k)}\right)\right)d\mathbf{x},
\end{aligned}
\end{equation}
defined in a domain $\Omega$, and take an infinitesimal variation of $\mathbf{u}^{(r)}$ given by
\begin{equation}
\mathbf{u}^{(r)}\longrightarrow\mathbf{u}^{(r)}+\delta\mathbf{u}^{(r)},
\end{equation}
\emph{i.e.}, by means of \eqref{expansion_u_der}, 
\begin{equation}
\mathbf{u}_{(k)}^{(r)}\longrightarrow\mathbf{u}_{(k)}^{(r)}+\delta\mathbf{u}_{(k)}^{(r)},\qquad k=0,\ldots,p,
\end{equation}
where $\delta\mathbf{u}_{(k)}^{(r)}=0$ for $k>0$, and $\delta\mathbf{u}^{(\ell)}_{(0)}=0$ for $\ell=0,\ldots,r-1$ on the boundary of $\Omega$. 

By requiring the 
first variation of the approximate Lagrangian action \eqref{action} to be $O(\varepsilon^{p+1})$ under variations of order $O(\varepsilon^{p+1})$ in $\partial\Omega$, we obtain
\begin{equation}
\label{first_variation}
\begin{aligned}
\delta \mathcal{J}&=\mathcal{J}\left(\mathbf{x},\mathbf{u}^{(r)}+\delta \mathbf{u}^{(r)};\varepsilon\right)-\mathcal{J}\left(\mathbf{x},\mathbf{u}^{(r)};\varepsilon\right)\approx\\
&\approx\int_{\Omega}\left(\sum_{k=0}^p\varepsilon^k\left(\widetilde{\mathcal{L}}_{(k)}\left(\mathbf{x},\mathbf{u}_{(0)}^{(r)}+\delta \mathbf{u}_{(0)}^{(r)},\mathbf{u}_{(1)}^{(r)},\ldots,\mathbf{u}_{(k)}^{(r)}\right)\right.\right.\\
&-\left.\left.\widetilde{\mathcal{L}}_{(k)}\left(\mathbf{x},\mathbf{u}_{(0)}^{(r)},\mathbf{u}_{(1)}^{(r)},\dots,\mathbf{u}_{(k)}^{(r)}\right)\right)\right)d\mathbf{x}=\\
&=\int_{\Omega}\left(\sum_{k=0}^p\varepsilon^k\left(\sum_{\alpha=1}^{m}\left(\frac{\partial\widetilde{\mathcal{L}}_{(k)}}{\partial u_{(0)\alpha}}\delta u_{(0)\alpha}+\sum_{i=1}^{n}\frac{\partial\widetilde{\mathcal{L}}_{(k)}}{\partial u_{(0)\alpha,i}}\delta u_{(0)\alpha,i}\right.\right.\right.\\
&\left.+\left.\left.\ldots+\sum_{i_1=1}^{n}\dots\sum_{i_r=i_{r-1}}^{n}\frac{\partial\widetilde{\mathcal{L}}_{(k)}}{\partial u_{(0)\alpha,i_1\ldots i_r}}\delta u_{(0)\alpha,i_1\ldots i_r}\right)\right)\right)d\mathbf{x}.
\end{aligned}
\end{equation}
Repeatedly integrating by parts, relation \eqref{first_variation} provides
\begin{equation}
\label{first_variation_2}
\delta \mathcal{J}\approx\int_{\Omega}\left(\sum_{\alpha=1}^{m}\left(E_{u_{(0)\alpha}}\left(\sum_{k=0}^{p}\varepsilon^k\widetilde{\mathcal{L}}_{(k)}\right)\delta u_{(0)\alpha}\right)+\sum_{i=1}^{n}D_i W^i\right)d\mathbf{x},
\end{equation}
where $E_{u_{(0)\alpha}}$ are the approximate Euler operators defined in \eqref{euler-approx}, and 
\begin{equation}
\begin{aligned}
W^i&\equiv W^i\left(\mathbf{x},\mathbf{u}^{(r)}_{(0)},\ldots,\mathbf{u}^{(r)}_{(k)},\delta\mathbf{u}_{(0)}^{(r)};\varepsilon\right)=\\
&=\sum_{k=0}^{p}\varepsilon^k\left(\sum_{\alpha=1}^{m}\left(\left(\frac{\partial \mathcal{\widetilde{L}}_{(k)}}{\partial u_{(0)\alpha,i}}+\ldots+(-1)^{r-1}\sum_{j_1=1}^{n}\ldots\right.\right.\right.\\ 
&\left.\ldots\sum_{j_{r-1}=j_{r-2}}^{n} D_{j_1\dots j_{r-1}}\frac{\partial \mathcal{\widetilde{L}}_{(k)}}{\partial u_{(0)\alpha,i j_1\dots j_{r-1}}}\right)\delta u_{(0)\alpha}\\
&+\sum_{j_1=1}^{n}\left(\frac{\partial \mathcal{\widetilde{L}}_{(k)}}{\partial u_{(0)\alpha,ij_1}}+\ldots+(-1)^{r-2}\sum_{j_2=j_{1}}^{n}\ldots\right.\\
&\left.\ldots \sum_{j_{r-1}=j_{r-2}}^{n}D_{j_2\dots j_{r-1}}\frac{\partial \mathcal{\widetilde{L}}_{(k)}}{\partial u_{(0)\alpha,i j_1\dots j_{r-1}}}\right)\delta u_{(0)\alpha,j_1}\\
&+\left.\left.\ldots+\sum_{j_1=1}^{n}\ldots \sum_{j_{r-1}=j_{r-2}}^{n}\frac{\partial \mathcal{\widetilde{L}}_{(k)}}{\partial u_{(0)\alpha,i j_1\dots j_{r-1}}}\delta u_{(0)\alpha,j_1\ldots j_{r-1}}\right)\right).
\end{aligned}
\end{equation}
Then, by applying the divergence theorem to relation \eqref{first_variation_2}, the first variation of the approximate Lagrangian action reads
\begin{equation}
\begin{aligned}
\delta \mathcal{J}&\approx\int_{\Omega}\left(\sum_{\alpha=1}^{m}\left(E_{u_{(0)\alpha}}\left(\sum_{k=0}^{p}\varepsilon^k\widetilde{\mathcal{L}}_{(k)}\right)\delta u_{(0)\alpha}\right)+\sum_{i=1}^{n}D_i W^i\right)d\mathbf{x}=\\
&=\int_{\Omega}\sum_{\alpha=1}^{m}\left(E_{u_{(0)\alpha}}\left(\sum_{k=0}^{p}\varepsilon^k\widetilde{\mathcal{L}}_{(k)}\right)\delta u_{(0)\alpha}\right)d\mathbf{x}+\int_{\partial\Omega}\sum_{i=1}^{n} W^i \phi_i dS,
\end{aligned}
\end{equation}
where $dS$ is the infinitesimal element of $\partial\Omega$, and $\boldsymbol{\phi}=(\phi_1,\ldots,\phi_n)$ the unit outward normal vector to $\partial\Omega$.

Since $\delta u_{(0)\alpha},\delta u_{(0)\alpha,j_1},\ldots,\delta u_{(0)\alpha,j_1\ldots j_{r-1}}$ vanish on $\partial\Omega$, we recover the approximate Euler--Lagrange equations:
\begin{equation}
\begin{aligned}
\sum_{k=0}^p&\varepsilon^k\left(\frac{\partial \widetilde{\mathcal{L}}_{(k)}}{\partial u_{(0)\alpha}}
-\sum_{i=1}^nD_{i}\left(\frac{\partial \widetilde{\mathcal{L}}_{(k)}
}{\partial u_{(0)\alpha,i}}\right)\right.\\
&\left.\quad+\ldots+(-1)^r\sum_{i_1=1}^n\ldots\sum_{i_r=i_{r-1}}^nD_{i_1}\dots D_{i_r}\left(\frac{\partial \widetilde{\mathcal{L}}_{(k)}
}{\partial u_{(0)\alpha,i_1\dots i_r}}\right)\right)\approx 0,
\end{aligned}
\end{equation}
with $\alpha=1,\ldots,m$.
\end{proof}

\begin{remark}
It is worth of being remarked that the computational approach for the approximate multipliers (\ref{approx-mult-us}) is different from the definitions (\ref{approx-mult-BGI}) and (\ref{approx-mult-FS}) in  Approach A \cite{Jamal-2019,Tara} and Approach B \cite{Jamal-2019}, respectively. In fact, in the proposed method, we expand the Lagrange multipliers according to \eqref{expansion_u_der}, \eqref{exp-1}, \eqref{exp-2}, as perturbation analysis requires, and introduce a new definition of approximate Euler operator (\ref{euler-approx}) that, compared to (\ref{euler-BGI}) and (\ref{euler-FS}), allows us to keeping all the elegant and not-expensive features of the Approach A (in particular, the number of the involved Euler operators is equal to the number of the dependent variables).
\end{remark}
Conditions \eqref{euler-cond} can be separated at each order of approximation in $\varepsilon$ and, due to the arbitrariness of the variables $\mathbf{u}^{(r)}_{(k)}$ $(k=0,\ldots,p)$, can be split into an overdetermined linear system of differential equations for the unknown  approximate multipliers. As a consequence, the linear combinations of the recovered approximate multipliers with the perturbed differential equations provide expressions in the approximate divergence form \eqref{divergence_form}. 
Last but not the least, it is not always possible to construct the approximate fluxes corresponding to a given set of approximate multipliers. In fact, this task can be faced as the problem of inversion of the divergence differential operator. There are several methods to determine the approximate fluxes starting from a known set of
approximate multipliers \cite{BlumanCheviakovAnco}.
Finally, it is worth of being underlined that, if the perturbed system \eqref{sysgen} arises from a variational principle, then the approximate direct method provides the approximate conservation laws derived by means of the approximate Noether procedure \cite{GO-Maths-2021}. In the next Section, we compare and discuss the results obtainable by using the different 
approaches considering a perturbed nonlinear diffusion equation.

\section{A comparison of the different methods}
\label{sec:4}
Let us consider the equation
\begin{equation}
\label{diff-eq}
\Delta(t,x,u,u_{,t},u_{,x},u_{,xx};\varepsilon)\equiv u_{,t}-(u^{-2}u_{,x})_{,x}-\varepsilon (u-u^{-1})_{,x}=0,
\end{equation}
where $u\equiv u(t,x;\varepsilon)$, $\varepsilon$ is the small parameter, and the subscripts denote partial derivatives,
and compute the first order approximate multipliers and the corresponding approximate conservation laws.

Equation \eqref{diff-eq} has already been investigated in \cite{Tara}, where
looking for the approximate multipliers $\Lambda(t,x,u;\varepsilon)$ of the form
\begin{equation}
\Lambda=\Lambda_{(0)}+\varepsilon\Lambda_{(1)},
\end{equation}
with $\Lambda_{(k)}\equiv\Lambda_{(k)}\left(t,x,u\right)$, $k=0,1$, and using
the second order Euler operator
\begin{equation}
\begin{aligned}
E_{u}&=\frac{\partial }{\partial u}
-D_{t}\left(\frac{\partial 
}{\partial u_{,t}}\right)-D_{x}\left(\frac{\partial 
}{\partial u_{,x}}\right)\\
&+D_{t}D_{t}\left(\frac{\partial 
}{\partial u_{,tt}}\right)+D_{t}D_{x}\left(\frac{\partial 
}{\partial u_{,tx}}\right)+D_{x}D_{x}\left(\frac{\partial 
}{\partial u_{,xx}}\right),
\end{aligned}
\end{equation}
it has been obtained the overdetermined system
\begin{equation}
E_{u}\left(\Lambda\Delta\right)=O(\varepsilon^2), 
\end{equation}
\emph{i.e.},
\begin{equation}
E_{u}\left((\Lambda_{(0)}+\varepsilon\Lambda_{(1)})(\Delta_{(0)}+\varepsilon\Delta_{(1)})\right)=O(\varepsilon^2).
\end{equation}
Then terms are separated at the various powers of $\varepsilon$, providing
\begin{equation}
\label{eu-BGI}
E_{u}\left(\Lambda_{(0)}\Delta_{(0)}\right)=0,\qquad E_{u}\left(\Lambda_{(0)}\Delta_{(1)}+\Lambda_{(1)}\Delta_{(0)}\right)=0,
\end{equation}
where
\begin{equation}
\Delta_{(0)}=u_{,t}-(u^{-2}u_{,x})_{,x},\qquad
\Delta_{(1)}=-(u-u^{-1})_{,x}.
\end{equation}
By solving \eqref{eu-BGI}, the following sets of approximate multipliers 
along with the associated approximate fluxes have been determined:
\begin{itemize}
\item
\begin{equation}
\label{results-BGI-1}
\Lambda^{1}=1,\qquad
\Phi^t_1=u,\qquad \Phi^x_1=-u^{-2}u_{,x} -\varepsilon\left(u-u^{-1}\right);
\end{equation}
\item
\begin{equation}
\label{results-BGI-2}
\begin{aligned}
&\Lambda^{2}=x + \varepsilon\left(t+ \frac{x^2}{2}\right),\qquad \Phi^t_2=xu + \varepsilon\left(t+ \frac{x^2}{2}\right)u,\\
&\Phi^x_2=-xu^{-2}u_{,x}- u^{-1} -\varepsilon\left(\left(t+ \frac{x^2}{2}\right)u^{-2}u_{,x}+xu\right).
\end{aligned}
\end{equation}
\end{itemize}
Also, the remaining approximate multipliers do not produce new independent approximate conservation laws because
\begin{equation}
\Lambda^{j+2}=\varepsilon\Lambda^{j},\qquad
\Phi^t_{j+2}=\varepsilon\Phi^t_j,\qquad \Phi^x_{j+2}=\varepsilon\Phi^x_j,\qquad j=1,2.
\end{equation}
Different results can be obtained by prior expansion of dependent variables. In fact,
substituting 
\begin{equation}
\label{expansion-example}
u(t,x;\varepsilon)=u_{(0)}(t,x)+\varepsilon u_{(1)}(t,x)+ O(\varepsilon^2),
\end{equation}
into \eqref{diff-eq}, and separating  terms at the various powers of $\varepsilon$, we obtain the coupled system
\begin{equation}
\begin{aligned}
&\Delta_{(0)}(t,x,u_{(0)},u_{(0),t},u_{(0),x},u_{(0),xx})=0\\
&\Delta_{(1)}(t,x,u_{(0)},u_{(1)},u_{(0),t},u_{(0),x},u_{(1),t},u_{(1),x},u_{(0),xx},
u_{(1),xx})=0,
\end{aligned}
\end{equation}
where
\begin{equation}
\label{pert-diff}
\begin{aligned}
\Delta_{(0)}&\equiv u_{(0),t}-u_{(0)}^{-2}u_{(0),xx}+2u_{(0)}^{-3}u_{(0),x}^2=0,\\
\Delta_{(1)}&\equiv u_{(1),t}+2u_{(0)}^{-3}u_{(1)}u_{(0),xx}-u_{(0)}^{-2}u_{(1),xx}-6u_{(0)}^{-4}u_{(1)}u_{(0),x}^2\\
&+4u_{(0)}^{-3}u_{(0),x}u_{(1),x}-(u_{(0)}^{-2}+1)u_{(0),x}=0.
\end{aligned}
\end{equation}
Thence,  the approximate multipliers of Eq. \eqref{diff-eq} are defined as the exact multipliers of system \eqref{pert-diff}; here, we look for multipliers $\Lambda_{(k)}$ ($k=0,1)$ of the form
\begin{equation}
\Lambda_{(0)}\equiv \Lambda_{(0)}(t,x,u_{(0)},u_{(1)}),\qquad
\Lambda_{(1)}\equiv \Lambda_{(1)}(t,x,u_{(0)},u_{(1)}).
\end{equation}
Introducing the Euler operators with respect to $u_{(0)}$ and $u_{(1)}$, say
\begin{equation}
\label{euler-FS-examp}
\begin{aligned}
E_{u_{(0)}}&=\frac{\partial }{\partial u_{(0)}}
-D_{t}\left(\frac{\partial 
}{\partial u_{(0),t}}\right)-D_{x}\left(\frac{\partial 
}{\partial u_{(0),x}}\right)\\
&+D_{t}D_{t}\left(\frac{\partial 
}{\partial u_{(0),tt}}\right)+D_{t}D_{x}\left(\frac{\partial 
}{\partial u_{(0),tx}}\right)+D_{x}D_{x}\left(\frac{\partial 
}{\partial u_{(0),xx}}\right),\\
E_{u_{(1)}}&=\frac{\partial }{\partial u_{(1)}}
-D_{t}\left(\frac{\partial 
}{\partial u_{(1),t}}\right)-D_{x}\left(\frac{\partial 
}{\partial u_{(1),x}}\right)\\
&+D_{t}D_{t}\left(\frac{\partial 
}{\partial u_{(1),tt}}\right)+D_{t}D_{x}\left(\frac{\partial 
}{\partial u_{(1),tx}}\right)+D_{x}D_{x}\left(\frac{\partial 
}{\partial u_{(1),xx}}\right),
\end{aligned}
\end{equation}
and requiring
\begin{equation}
\label{deteqs-FS}
E_{u_{(0)}}(\Lambda_{(0)}\Delta_{(0)}+\Lambda_{(1)}\Delta_{(1)})=0,\qquad E_{u_{(1)}}(\Lambda_{(0)}\Delta_{(0)}+\Lambda_{(1)}\Delta_{(1)})=0,
\end{equation}
we obtain the following multipliers for \eqref{pert-diff} along with the associated fluxes:
\begin{itemize}
\item 
\begin{equation}
\label{results-FS-1}
\begin{aligned}
&\Lambda^{1}_{(0)}=t+\frac{x^2}{2},\qquad \Lambda^{1}_{(1)}=x,\\ &\Phi^t_1=\left(t+\frac{x^2}{2}\right)u_{(0)}+xu_{(1)},\\
&\Phi^x_1=\left(-\left(t+\frac{x^2}{2}\right)u_{(0)}^{-2}+2xu_{(0)}^{-3}u_{(1)}\right)u_{(0),x}\\
&\phantom{\Phi^x_1}-xu_{(0)}^{-2}u_{(1),x}-xu_{(0)}+u_{(0)}^{-2}u_{(1)};
\end{aligned}
\end{equation}
\item 
\begin{equation}
\label{results-FS-2}
\begin{aligned}
&\Lambda^{2}_{(0)}=x,\qquad \Lambda^{2}_{(1)}=0,\\
&\Phi^t_2=xu_{(0)},\qquad \Phi^x_2=-\left(x u_{(0)}^{-2}u_{(0),x}+u_{(0)}^{-1}\right);
\end{aligned}
\end{equation}

\item 
\begin{equation}
\label{results-FS-3}
\begin{aligned}
&\Lambda^{3}_{(0)}=1,\qquad \Lambda^{3}_{(1)}=0,\\
&\Phi^t_3=u_{(0)},\qquad
\Phi^x_3=- u_{(0)}^{-2}u_{(0),x};
\end{aligned}
\end{equation}

\item 
\begin{equation}
\label{results-FS-4}
\begin{aligned}
&\Lambda^{4}_{(0)}=0,\qquad \Lambda^{4}_{(1)}=1,\\
&\Phi^t_4=u_{(1)},\qquad \Phi^x_4=2u_{(0)}^{-3}u_{(1)}u_{(0),x}-u_{(0)}^{-2}u_{(1),x}-u_{(0)}+u_{(0)}^{-1}.
\end{aligned}
\end{equation}
\end{itemize}


Now, let us apply our direct procedure for approximate conservation laws. Therefore, let us use the expansion \eqref{expansion-example} at first order in $\varepsilon$, and look for approximate multipliers $\Lambda(t,x,u;\varepsilon)$ of the form
\begin{equation}
\begin{aligned}
\Lambda&=\Lambda_{(0)}+\varepsilon\left(\Lambda_{(1)}+
\frac{\partial \Lambda_{(0)}}{\partial u_{(0)}}u_{(1)} \right),
\end{aligned}
\end{equation}
where $\Lambda_{(k)}\equiv\Lambda_{(k)}\left(t,x,u_{(0)}\right)$, $k=0,1$.

By using Theorem~\ref{teo1}, we obtain the overdetermined system
\begin{equation}
E_{u_{(0)}}\left(\Lambda\Delta\right)=O(\varepsilon^2), 
\end{equation}
\emph{i.e.},
\begin{equation}\label{deteqs-approx-example}
E_{u_{(0)}}\left(\left(\Lambda_{(0)}+\varepsilon\left(\Lambda_{(1)}+
\frac{\partial \Lambda_{(0)}}{\partial u_{(0)}}u_{(1)} \right)\right)\left(\Delta_{(0)}+\varepsilon\Delta_{(1)}\right)\right)=O(\varepsilon^2), 
\end{equation}
with $\Delta_{(0)}$ and $\Delta_{(1)}$ given in \eqref{pert-diff}.
Differently from \eqref{euler-FS-examp}, the only one approximate Euler operator to be used is
\begin{equation}
\begin{aligned}
E_{u_{(0)}}&=\frac{\partial }{\partial u_{(0)}}
-D_{t}\left(\frac{\partial 
}{\partial u_{(0),t}}\right)-D_{x}\left(\frac{\partial 
}{\partial u_{(0),x}}\right)\\
&+D_{t}D_{t}\left(\frac{\partial 
}{\partial u_{(0),tt}}\right)+D_{t}D_{x}\left(\frac{\partial 
}{\partial u_{(0),tx}}\right)+D_{x}D_{x}\left(\frac{\partial 
}{\partial u_{(0),xx}}\right).
\end{aligned}
\end{equation}
Then, separation of terms of system \eqref{deteqs-approx-example} at the various powers of $\varepsilon$ provides
\begin{equation}
\label{deteqs-approx}
\begin{aligned}
&E_{u_{(0)}}\left(\Lambda_{(0)}\Delta_{(0)}\right)=0,\\ 
&E_{u_{(0)}}\left(\Lambda_{(0)}\Delta_{(1)}+\left(\Lambda_{(1)}+
\frac{\partial \Lambda_{(0)}}{\partial u_{(0)}}u_{(1)}\right)\Delta_{(0)}\right)=0;
\end{aligned}
\end{equation}
by solving \eqref{deteqs-approx}, we determine the following sets of approximate multipliers 
along with the corresponding approximate fluxes:
\begin{itemize}
\item
\begin{equation}
\label{results-approx-1}
\begin{aligned}
&\Lambda^{1}=1,\qquad
\Phi^t_1=u_{(0)}+\varepsilon u_{(1)},\\
&\Phi^x_1=-u_{(0)}^{-2}u_{(0),x} +\varepsilon\left(2u_{(0)}^{-3}u_{(1)}u_{(0),x}-u_{(0)}^{-2}u_{(1),x}-u_{(0)}+u_{(0)}^{-1}\right);
\end{aligned}
\end{equation}
\item
\begin{equation}
\label{results-approx-2}
\begin{aligned}
\Lambda^{2}&=x + \varepsilon\left(t+ \frac{x^2}{2}\right),\\
\Phi^t_2&=xu_{(0)} + \varepsilon\left(\left(t+ \frac{x^2}{2}\right)u_{(0)}+x u_{(1)}\right),\\
\Phi^x_2&=-xu_{(0)}^{-2}u_{(0),x}- u_{(0)}^{-1} \\
&-\varepsilon\left(\left(\left(t+ \frac{x^2}{2}\right)u_{(0)}^{-2}-2xu_{(0)}^{-3}u_{(1)}\right)u_{(0),x}\right.\\
&\left.+\phantom{\frac{1}{2}}xu_{(0)}^{-2}u_{(1),x}+xu_{(0)}-u_{(0)}^{-2}u_{(1)}\right).
\end{aligned}
\end{equation}
\end{itemize}
The remaining approximate multipliers do not produce new independent approximate conservation laws; in fact, we obtain:
\begin{equation}
\Lambda^{j+2}=\varepsilon\Lambda^{j},\qquad
\Phi^t_{j+2}=\varepsilon\Phi^t_j,\qquad \Phi^x_{j+2}=\varepsilon\Phi^x_j,\qquad j=1,2.
\end{equation}
\begin{remark}
We can compare the results obtained with the different methods. We note that the approximate fluxes \eqref{results-approx-1}--\eqref{results-approx-2}, recovered by means of the consistent direct approach introduced in this paper, are the expansion of \eqref{results-BGI-1}--\eqref{results-BGI-2} obtained by the Approach A; this allows to obtain a consistent expansion of the involved quantities, with the benefit that the amount of algebra is not increased; however, results \eqref{results-FS-1}--\eqref{results-FS-4}, determined by Approach B, are different and require a higher computational cost. 
\end{remark}

\section{Applications}\label{sec:5}
In this Section, we present some examples of physical interest not admitting a Lagrangian formulation, and apply the consistent direct method for approximate conservation laws. By limiting ourselves to a first order approximation, we are able to determine the approximate multipliers and the corresponding approximate conservation laws. All the required computations  have been done by means of the general and freely available package ReLie~\cite{Oliveri-relie} written in the Computer Algebra System Reduce~\cite{Reduce}.

\subsection{The perturbed KdV--Burgers equation}
Let us consider the perturbed KdV--Burgers equation
\begin{equation}
\Delta(t,x,u,u_{,t},u_{,x},u_{,xx},u_{,xxx};\varepsilon)\equiv u_{,t}+uu_{,x}+u_{,xxx}-\varepsilon u_{,xx}=0,
\end{equation}
where $u\equiv u(t,x;\varepsilon)$, and $\varepsilon$ is the small parameter.

For this equation, approximate multipliers and approximate conservation laws have been determined in \cite{Jamal-2019}, following the Approach A, which does not require the expansion of the dependent variable and makes use of the Euler operator \eqref{euler-BGI}.

Therefore, let us expand $u(t,x;\varepsilon)$ at first order in $\varepsilon$, {i.e.},
\begin{equation}
u(t,x;\varepsilon)=u_{(0)}(t,x)+\varepsilon u_{(1)}(t,x)+ O(\varepsilon^2),
\end{equation}
and look for approximate multipliers $\Lambda(t,x,u,u_{,x},u_{,xx};\varepsilon)$, depending on second order derivatives, of the form
\begin{equation}
\begin{aligned}
\Lambda&=\Lambda_{(0)}+\varepsilon\left(\Lambda_{(1)}+
\frac{\partial \Lambda_{(0)}}{\partial u_{(0)}}u_{(1)} +
\frac{\partial \Lambda_{(0)}}{\partial {u}_{(0),x}}{u}_{(1),x}+\frac{\partial \Lambda_{(0)}}{\partial {u}_{(0),xx}}{u}_{(1),xx}\right),
\end{aligned}
\end{equation}
where $\Lambda_{(k)}\equiv\Lambda_{(k)}\left(t,x,u_{(0)},u_{(0),x},u_{(0),xx}\right)$, $k=0,1$.

By using Theorem \ref{teo1}, we obtain the overdetermined system
\begin{equation}
E_{u_{(0)}}\left(\Lambda\Delta\right)=O(\varepsilon^2), 
\end{equation}
where
\begin{equation}
\begin{aligned}
E_{u_{(0)}}&=\frac{\partial }{\partial u_{(0)}}
-D_{t}\left(\frac{\partial 
}{\partial u_{(0),t}}\right)-D_{x}\left(\frac{\partial 
}{\partial u_{(0),x}}\right)\\
&+D_{t}D_{t}\left(\frac{\partial 
}{\partial u_{(0),tt}}\right)+D_{t}D_{x}\left(\frac{\partial 
}{\partial u_{(0),tx}}\right)+D_{x}D_{x}\left(\frac{\partial 
}{\partial u_{(0),xx}}\right)\\
&-D_{t}D_{t}D_{t}\left(\frac{\partial 
}{\partial u_{(0),ttt}}\right)-D_{t}D_{t}D_{x}\left(\frac{\partial 
}{\partial u_{(0),ttx}}\right)\\
&-D_{t}D_{x}D_{x}\left(\frac{\partial 
}{\partial u_{(0),txx}}\right)-D_{x}D_{x}D_{x}\left(\frac{\partial 
}{\partial u_{(0),xxx}}\right),
\end{aligned}
\end{equation}
that, separating terms at the various powers of $\varepsilon$, gives
\begin{equation}
\label{eu1}
\begin{aligned}
&E_{u_{(0)}}\left(\Lambda_{(0)}\Delta_{(0)}\right)=0,\\ 
&E_{u_{(0)}}\left(\Lambda_{(0)}\Delta_{(1)}+\left(\Lambda_{(1)}+
\frac{\partial \Lambda_{(0)}}{\partial u_{(0)}}u_{(1)}\right.\right.\\
&\left.\left.\qquad+\frac{\partial \Lambda_{(0)}}{\partial {u}_{(0),x}}{u}_{(1),x}+\frac{\partial \Lambda_{(0)}}{\partial {u}_{(0),xx}}{u}_{(1),xx}\right)\Delta_{(0)}\right)=0.
\end{aligned}
\end{equation}
By solving \eqref{eu1}, we determine the following sets of approximate multipliers 
along with the corresponding approximate fluxes:
\begin{itemize}
\item
\begin{equation}
\Lambda^{1}=1,
\end{equation}
with
\begin{equation}
\begin{aligned}
\Phi^t_1&=(tu_{(0)}-x)u_{(0),x}+\varepsilon\left(tu_{(1)}u_{(0),x}+(tu_{(0)}-x)u_{(1),x}\right),\\
\Phi^x_1&=u_{(0),xx}-(tu_{(0)}-x)u_{(0),t}\\
&+\varepsilon\left(u_{(1),xx}-tu_{(1)}u_{(0),t}-u_{(0),x}-(tu_{(0)}-x)u_{(1),t}\right);
\end{aligned}
\end{equation}
\item
\begin{equation}
\Lambda^{2}=tu_{(0)}-x+\varepsilon\left(2t^2u_{(0),xx}+(tu_{(0)}-x)^2+t{u}_{(1)}\right),
\end{equation}
with
\begin{equation}
\begin{aligned}
\Phi^t_2&=\frac{(tu_{(0)}-x)^2}{2}u_{(0),x}-\varepsilon\left(t^2u_{(0),x}^2-\frac{(tu_{(0)}-x)^2}{2}u_{(1),x}\right.\\
&-\left.\left(\frac{(tu_{(0)}-x)^3}{3}+t((tu_{(0)}-x)u_{(1)}+1)\right)u_{(0),x}  \right),\\
\Phi^x_2&=(t u_{(0)}-x) u_{(0),xx}-\frac{t}{2}u_{(0),x}^2-\frac{(t u_{(0)}-x)^2}{2} u_{(0),t}+u_{(0),x}\\
&+\varepsilon\left(t^2u_{(0),xx}^2+((tu_{(0)}-x)^2+t u_{(1)}) u_{(0),xx}\phantom{\frac{1}{2}}\right.\\
&+\left.(t u_{(0)}-x) u_{(1),xx}+t(2tu_{(0),t}+xu_{(0),x}-u_{(1),x})u_{(0),x}\phantom{\frac{1}{2}}\right.\\
&-\left(\frac{(tu_{(0)}-x)^3}{3}+t((tu_{(0)}-x)u_{(1)}+1)\right)u_{(0),t}\\
&+\left.(t u_{(0)}-x)u_{(0),x}-\frac{(tu_{(0)}-x)^2}{2} u_{(1),t}+u_{(1),x}\right);
\end{aligned}
\end{equation}
\item
\begin{equation}
\Lambda^{3}=u_{(0)}+\varepsilon\left(4tu_{(0),xx}+2(tu_{(0)}-x)u_{(0)}+u_{(1)}\right),
\end{equation}
with
\begin{equation}
\begin{aligned}
\Phi^t_3&=(tu_{(0)}-x)u_{(0)}u_{(0),x}-\varepsilon\left(2tu_{(0),x}^2-(tu_{(0)}-x)u_{(0)}u_{(1),x}\right.\\
&-\left.((tu_{(0)}-x)^2u_{(0)}+(2tu_{(0)}-x)u_{(1)})u_{(0),x}\right),\\
\Phi^x_3&=u_{(0)}u_{(0),xx}-\frac{u_{(0),x}^2}{2}-(tu_{(0)}-x)u_{(0)}u_{(0),t}\\
&+\varepsilon\left(2tu_{(0),xx}^2+(2(t u_{(0)}-x)u_{(0)}+u_{(1)})u_{(0),xx}+u_{(0)}u_{(1),xx}\right.\\
&+(4tu_{(0),t}+xu_{(0),x}-u_{(1),x})u_{(0),x}-(tu_{(0)}-x) u_{(0)} u_{(1),t}\\
&-\left.((tu_{(0)}-x)^2u_{(0)}+(2tu_{(0)}-x)u_{(1)}) u_{(0),t}+u_{(0)}u_{(0),x}\right);
\end{aligned}
\end{equation}
\item
\begin{equation}
\Lambda^{4}=\varepsilon\left(u_{(0),xx}+\frac{u_{(0)}^2}{2}\right),
\end{equation}
with
\begin{equation}
\begin{aligned}
\Phi^t_4&=\varepsilon\left(-u_{(0),x}+(tu_{(0)}-x)u_{(0)}^2\right)u_{(0),x},\\
\Phi^x_4&=\varepsilon\left((u_{(0),xx}+u_{(0)}^2)u_{(0),xx}+(2u_{(0),x}-(tu_{(0)}-x)u_{(0)}^2)u_{(0),t}\right).
\end{aligned}
\end{equation}
\end{itemize}
The remaining approximate multipliers do not produce new independent approximate conservation laws; in fact, we obtain:
\begin{equation}
\Lambda^{j+4}=\varepsilon\Lambda^{j},
\end{equation}
with
\begin{equation}
\Phi^t_{j+4}=\varepsilon\Phi^t_j,\qquad \Phi^x_{j+4}=\varepsilon\Phi^x_j,\qquad j=1,2,3.
\end{equation}
\begin{remark}
The approximate multipliers above recovered are just the expansion of those
determined in \cite{Jamal-2019} by means of Approach A; however,  the associated approximate fluxes obtained in this paper are different from those reported in \cite{Jamal-2019}. 
\end{remark}

\subsection{A nonlinear wave equation}
Let us consider the perturbed nonlinear wave equation
\begin{equation}
	\Delta(t,x,u,u_{,tt},u_{,xx};\varepsilon)\equiv u_{,xx}-\frac{1}{c^2}u_{,tt}-\lambda u^3-\varepsilon f(u)=0,
\end{equation}
where $u\equiv u(t,x;\varepsilon)$, and $f(u)$ is an arbitrary function of its argument.
In \cite{Jamal-2019}, approximate multipliers have been determined following the Approach B, \emph{i.e.}, by requiring the expansion of the dependent variable and using the Euler operators defined in \eqref{euler-FS}.

Following the approach developed in this paper, let us expand the dependent variable $u(t,x;\varepsilon)$ at ﬁrst order in $\varepsilon$, and look for approximate multipliers $\Lambda(t,x,u,u_{,t},u_{,x};\varepsilon)$, depending on first order derivatives, of the form
\begin{equation}
\Lambda=\Lambda_{(0)}+\varepsilon\left(\Lambda_{(1)}+
\frac{\partial \Lambda_{(0)}}{\partial u_{(0)}}u_{(1)} +
\frac{\partial \Lambda_{(0)}}{\partial {u}_{(0),t}}{u}_{(1),t}+\frac{\partial \Lambda_{(0)}}{\partial {u}_{(0),x}}{u}_{(1),x}\right),
\end{equation}
with the positions
\begin{equation}
\Lambda_{(k)}=g^0_{(k)}+g^1_{(k)}u_{(0),t}+g^2_{(k)}u_{(0),x},\qquad k=0,1
\end{equation}
where $g^0_{(k)}\equiv g^0_{(k)}(t,x,u_{(0)})$ ($k=0,1$) are functions of their arguments.

By applying the approximate direct procedure, we split the system
\begin{equation}
E_{u_{(0)}}\left(\Lambda \Delta\right)=O(\varepsilon^2)
\end{equation}
at the various powers of $\varepsilon$, and we are able to recover the following sets of approximate multipliers and the associated approximate fluxes:
\begin{itemize}
\item
\begin{equation}
\Lambda^{1}=u_{(0),t}+\varepsilon{u}_{(1),t},
\end{equation}
with
\begin{equation}
\begin{aligned}
\Phi^t_1&=-\frac{1}{2c^2}u_{(0),t}^2-\frac{u_{(0),x}^2}{2}+\lambda x u_{(0)}^3u_{(0),x}\\
&-\varepsilon\left(\frac{x}{c^2}(u_{(0),t}u_{(1),t}+c^2u_{(0),x}u_{(1),x})+t(u_{(0),t}u_{(1),x}+u_{(0),x}u_{(1),t})\right.\\
&+\frac{1}{2}\left.(c^2t^2-x^2)\phantom{\frac{}{}}((3\lambda u_{(0)}^2u_{(1)}+f(u_{(0)}) )u_{(0),x}+\lambda u_{(0)}^3u_{(1),x})\right),\\
\Phi^x_1&=(u_{(0),x}-\lambda x u_{(0)}^3)u_{(0),t}\\
&+\varepsilon\left(u_{(0),t}u_{(1),x}+u_{(0),x}u_{(1),t}\phantom{\frac{}{}}\right.\\
&-\left.\phantom{\frac{}{}}(3\lambda u_{(0)}^2u_{(1)}+f(u_{(0)}))x u_{(0),t}-\lambda x u_{(0)}^3u_{(1),t}\right);
\end{aligned}
\end{equation}
\item
\begin{equation}
\Lambda^{2}=u_{(0),x}+\varepsilon{u}_{(1),x},
\end{equation}
with
\begin{equation}
\begin{aligned}
\Phi^t_2&=-\frac{1}{c^2}(u_{(0),t}+\lambda c^2 tu_{(0)}^3)u_{(0),x}\\
&-\varepsilon\left(\frac{1}{c^2}(u_{(0),t}u_{(1),x}+u_{(0),x}u_{(1),t})\right.\\
&+\left.(3\lambda u_{(0)}^2u_{(1)} + f(u_{(0)}))t u_{(0),x} +\lambda t u_{(0)}^3u_{(1),x}\phantom{\frac{}{}}\right),\\
\Phi^x_2&=\frac{1}{2c^2}u_{(0),t}^2+\frac{u_{(0),x}^2}{2}+\lambda t u_{(0)}^3u_{(0),t}\\
&+\varepsilon\left(\frac{1}{c^2}(u_{(0),t}u_{(1),t}-c^2 u_{(0),x}u_{(1),x})\right.\\
&+\left.\left(3\lambda u_{(0)}^2u_{(1)}+f(u_{(0)}) \right)t u_{(0),t}+\lambda t u_{(0)}^3u_{(1),t}\phantom{\frac{}{}}\right);
\end{aligned}
\end{equation}
\item
\begin{equation}
\Lambda^{3}=xu_{(0),t}+c^2 t u_{(0),x}+\varepsilon(xu_{(1),t}+c^2 t u_{(1),x}),
\end{equation}
with
\begin{align}
\Phi^t_3&=-\frac{x}{2c^2}u_{(0),t}^2-tu_{(0),t}u_{(0),x}-\frac{x}{2}u_{(0),x}^2 -\frac{\lambda(c^2t^2-x^2)}{2}u_{(0)}^3u_{(0),x}\nonumber\allowdisplaybreaks\\
&-\varepsilon\left(\frac{x}{c^2}(u_{(0),t}u_{(1),t}+c^2u_{(0),x}u_{(1),x})+t(u_{(0),t}u_{(1),x}+u_{(0),x}u_{(1),t})\right.\nonumber\allowdisplaybreaks\\
&+\frac{1}{2}\left.(c^2t^2-x^2)\phantom{\frac{}{}}((3\lambda u_{(0)}^2u_{(1)}+f(u_{(0)}) )u_{(0),x}+\lambda u_{(0)}^3u_{(1),x})\right),\\
\Phi^x_3&=\frac{t}{2}u_{(0),t}^2+xu_{(0),t}u_{(0),x}+\frac{c^2}{2}tu_{(0),x}^2 +\frac{\lambda(c^2t^2-x^2)}{2}u_{(0)}^3u_{(0),t}\nonumber\allowdisplaybreaks\\
&+\varepsilon\left(t(u_{(0),t}u_{(1),t}+c^2u_{(0),x}u_{(1),x})+x(u_{(0),t}u_{(1),x}+u_{(0),x}u_{(1),t})\phantom{\frac{}{}}\phantom{\frac{c^2}{c^2}}\right.\nonumber\allowdisplaybreaks\\
&+\frac{1}{2}\left.(c^2t^2-x^2)\phantom{\frac{}{}}((3\lambda u_{(0)}^2u_{(1)}+f(u_{(0)}) )u_{(0),t}+\lambda u_{(0)}^3u_{(1),t})\right).\nonumber\allowdisplaybreaks
\end{align}
\end{itemize}
The remaining approximate multipliers do not produce new independent approximate conservation laws; in fact, we obtain:
\begin{equation}
\Lambda^{j+3}=\varepsilon\Lambda^{j},
\end{equation}
with
\begin{equation}
\begin{aligned}
\Phi^t_{j+3}&=\varepsilon\Phi^t_j,\qquad \Phi^x_{j+3}=\varepsilon\Phi^x_j,\qquad j=1,2,3.
\end{aligned}
\end{equation}
\begin{remark}
In this case, the recovered approximate multipliers are different
from the ones reported in \cite{Jamal-2019}, and this is a consequence of the different approach.
\end{remark}
\subsection{Perturbed nonlinear second order Schr\"{o}dinger equation}
Let us consider the perturbed nonlinear second order Schr\"{o}dinger equation
\begin{equation}\label{sch1}
\ii p_{,t}+p_{,xx}+2|p|^2p-\varepsilon |p|^4 p=0,
\end{equation}
where $p\equiv p(t,x;\varepsilon)$ is the  complex--valued envelope of the wave.
By decomposing Eq. \eqref{sch1} into real and imaginary parts, we obtain the following system of partial differential equations:
\begin{equation}
\begin{aligned}
\Delta^1&\equiv u_{,t}+v_{,xx}+2v(u^2+v^2)-\varepsilon v\left(u^2+v^2\right)^2=0,\\
\Delta^2&\equiv v_{,t}-u_{,xx}-2u(u^2+v^2)+\varepsilon u\left(u^2+v^2\right)^2=0.
\end{aligned}
\end{equation}

Let us expand $u(t,x;\varepsilon)$ and $v(t,x;\varepsilon)$ at first order in $\varepsilon$, {i.e.},
\begin{equation}
\begin{aligned}
&u(t,x;\varepsilon)=u_{(0)}(t,x)+\varepsilon u_{(1)}(t,x)+ O(\varepsilon^2),\\
&v(t,x;\varepsilon)=v_{(0)}(t,x)+\varepsilon v_{(1)}(t,x)+ O(\varepsilon^2),
\end{aligned}
\end{equation}
and look for the approximate multipliers $\Lambda^\nu(t,x,u,v,u_{,x},v_{,x},u_{,xx},v_{,xx};\varepsilon)$ ($\nu=1,2$) depending on second order derivatives, say
\begin{equation}
\label{second-order-mult}
\begin{aligned}
\Lambda^\nu&=\Lambda^\nu_{(0)}+\varepsilon\left(\Lambda^\nu_{(1)}+
\frac{\partial \Lambda^\nu_{(0)}}{\partial u_{(0)}}u_{(1)} +
\frac{\partial \Lambda^\nu_{(0)}}{\partial v_{(0)}}v_{(1)}+\frac{\partial \Lambda^\nu_{(0)}}{\partial {u}_{(0),x}}{u}_{(1),x}\right.\\
&+\left.\frac{\partial \Lambda^\nu_{(0)}}{\partial {v}_{(0),x}}{v}_{(1),x}+\frac{\partial \Lambda^\nu_{(0)}}{\partial {u}_{(0),xx}}{u}_{(1),xx}+\frac{\partial \Lambda^\nu_{(0)}}{\partial {v}_{(0),xx}}{v}_{(1),xx}\right),
\end{aligned}
\end{equation}
where $\Lambda^\nu_{(k)}\equiv\Lambda^\nu_{(k)}\left(t,x,u_{(0)},v_{(0)},u_{(0),x},v_{(0),x},u_{(0),xx},v_{(0),xx}\right)$, $\;k=0,1$.

From Theorem \ref{teo1}, we have
\begin{equation}
\begin{aligned}
&E_{u_{(0)}}\left(\Lambda^1\Delta^1+\Lambda^2\Delta^2\right)=0,\\
&E_{v_{(0)}}\left(\Lambda^1\Delta^1+\Lambda^2\Delta^2\right)=0,
\end{aligned}
\end{equation}
and the following sets of approximate multipliers with the associated approximate fluxes are recovered:
\begin{itemize}
\item
\begin{equation}
\begin{aligned}
\Lambda^1_1&=v_{(0),xx}+2v_{(0)}(u_{(0)}^2+v_{(0)}^2)+\varepsilon\left(v_{(1),xx}-v_{(0)}(u_{(0)}^2+v_{(0)}^2)^2\right.\\
&+\left.2(u_{(0)}^2v_{(1)}+2u_{(0)}u_{(1)}v_{(0)}+3v^2_{(0)}v_{(1)})\right),\\
\Lambda^2_1&=-u_{(0),xx}-2u_{(0)}(u_{(0)}^2+v_{(0)}^2)-\varepsilon\left(u_{(1),xx}-u_{(0)}(u_{(0)}^2+v_{(0)}^2)^2\right.\\
&+\left.2(v^2_{(0)}u_{(1)}+2u_{(0)}v_{(0)}v_{(1)}+3u_{(0)}^2u_{(1)})\right),
\end{aligned}
\end{equation}
with
\begin{equation}
\begin{aligned}
\Phi^t_1&=\frac{1}{2}\left(u_{(0),x}^2+v_{(0),x}^2-(u_{(0)}^2+v_{(0)}^2)^2\right)\\
&+\varepsilon\left(\left(u_{(1),x}-x u_{(0)}(u_{(0)}^2+v_{(0)}^2)^2\right)u_{(0),x}\right.\\
&+\left(v_{(1),x}-x v_{(0)}(u_{(0)}^2+v_{(0)}^2)^2\right)v_{(0),x}\\
&-\left.2(u_{(0)}^2+v_{(0)}^2)(u_{(0)}u_{(1)}+v_{(0)}v_{(1)})\right),\\
\Phi^x_1&=-(u_{(0),t}u_{(0),x}+v_{(0),t}v_{(0),x})\\
&-\varepsilon\left(\left(u_{(1),x}-x u_{(0)}(u_{(0)}^2+v_{(0)}^2)^2\right)u_{(0),t}\right.\\
&+\left.\left(v_{(1),x}-x v_{(0)}(u_{(0)}^2+v_{(0)}^2)^2\right)v_{(0),t}+u_{(0),x}u_{(1),t}+v_{(0),x}v_{(1),t}\right);
\end{aligned}
\end{equation}
\item
\begin{equation}
\begin{aligned}
\Lambda^1_2&=2tu_{(0),x}+x v_{(0)}+\varepsilon\left(2tu_{(1),x}+x v_{(1)}\right),\\
\Lambda^2_2&=2tv_{(0),x}-x u_{(0)}+\varepsilon\left(2tv_{(1),x}-x u_{(1)}\right),
\end{aligned}
\end{equation}
with
\begin{equation}
\begin{aligned}
\Phi^t_2&=2tv_{(0)}u_{(0),x}+\frac{x}{2}(u_{(0)}^2+v_{(0)}^2)\\
&+\varepsilon\left(t(tu_{(0)}(u_{(0)}^2+v_{(0)}^2)^2+2v_{(1)})u_{(0),x}+x(u_{(0)}u_{(1)}+v_{(0)}v_{(1)})\right.\\
&+\left.t^2v_{(0)}(u_{(0)}^2+v_{(0)}^2)^2v_{(0),x}+2tv_{(0)}u_{(1),x}\right),\\
\Phi^x_2&=-t\left(u_{(0),x}^2+v_{(0),x}^2+2v_{(0)}u_{(0),t}+(u_{(0)}^2+v_{(0)}^2)^2\right)\\
&-x(v_{(0)}u_{(0),x}-u_{(0)}v_{(0),x})-u_{(0)}v_{(0)}\\
&-\varepsilon\left(t(tu_{(0)}(u_{(0)}^2+v_{(0)}^2)^2+2v_{(1)})u_{(0),t}+2tv_{(0)}u_{(1),t}\right.\\
&+\left.(2tu_{(1),x}+xv_{(1)})u_{(0),x}+(2tv_{(1),x}-xu_{(1)})v_{(0),x}\phantom{\frac{}{}}\right.\\
&+t^2v_{(0)}(u_{(0)}^2+v_{(0)}^2)^2v_{(0),t}+x(v_{(0)}u_{(1),x}-u_{(0)}v_{(1),x})\\
&+\left.4t(u_{(0)}^2+v_{(0)}^2)(u_{(0)}u_{(1)}+v_{(0)}v_{(1)})+u_{(0)}v_{(1)}+v_{(0)}u_{(1)}\right);
\end{aligned}
\end{equation}
\item
\begin{equation}
\begin{aligned}
\Lambda^1_3&=v_{(0)}+\varepsilon v_{(1)},\qquad
\Lambda^2_3=-u_{(0)}-\varepsilon u_{(1)},
\end{aligned}
\end{equation}
with
\begin{equation}
\begin{aligned}
\Phi^t_3&=\frac{1}{2}(u_{(0)}^2+v_{(0)}^2)+\varepsilon(u_{(0)}u_{(1)}+v_{(0)}v_{(1)}),\\
\Phi^x_3&=u_{(0)}v_{(0),x}-v_{(0)}u_{(0),x}\\
&+\varepsilon\left(u_{(1)}v_{(0),x}-v_{(1)}u_{(0),x}+u_{(0)}v_{(1),x}-v_{(0)}u_{(1),x}\right);
\end{aligned}
\end{equation}
\item
\begin{equation}
\begin{aligned}
\Lambda^1_4&=u_{(0),x}+\varepsilon u_{(1),x},\qquad
\Lambda^2_4=v_{(0),x}+\varepsilon v_{(1),x},
\end{aligned}
\end{equation}
with
\begin{equation}
\begin{aligned}
\Phi^t_4&=v_{(0)}u_{(0),x}+\varepsilon\left((tu_{(0)}(u_{(0)}^2+v_{(0)}^2)^2+v_{(1)})u_{(0),x}\right.\\
&+\left.tv_{(0)}(u_{(0)}^2+v_{(0)}^2)^2v_{(0),x}+v_{(0)}u_{(1),x}\right),\\
\Phi^x_4&=-\frac{1}{2}\left(u_{(0),x}^2+v_{(0),x}^2+(u_{(0)}^2+v_{(0)}^2)^2\right)-v_{(0)}u_{(0),t}\\
&-\varepsilon\left(u_{(0),x}u_{(1),x}+v_{(0),x}v_{(1),x}+(tu_{(0)}(u_{(0)}^2+v_{(0)}^2)^2+v_{(1)})u_{(0),t}\right.\\
&+tv_{(0)}(u_{(0)}^2+v_{(0)}^2)^2v_{(0),t}+v_{(0)}u_{(1),t}\\
&+\left.2(u_{(0)}^2+v_{(0)}^2)(u_{(0)}u_{(1)}+v_{(0)}v_{(1)})\right).
\end{aligned}
\end{equation}
\end{itemize}
The remaining approximate multipliers do not produce new independent approximate conservation laws; in fact, we obtain:
\begin{equation}
\Lambda^i_{j+4}=\varepsilon \Lambda^i_j,
\end{equation}
with
\begin{equation}
\Phi^t_{j+4}=\varepsilon\Phi^t_j,\qquad \Phi^x_{j+4}=\varepsilon\Phi^x_j,
\end{equation}
where $i=1,2$ and $j=1,\dots,4$.
\subsection{Perturbed nonlinear third order Schr\"{o}dinger equation}
Let us consider the  perturbed nonlinear third order Schr\"{o}dinger equation
\begin{equation}\label{sch2}
\ii p_{,t}+\frac{1}{2}p_{,xx}+|p|^2p+\ii\varepsilon\left(\beta_1 p_{,xxx}+\beta_2 |p|^2p_{,x}+\beta_3 p(|p|^2)_{,x}\right)=0,
\end{equation}
where $p\equiv p(t,x;\varepsilon)$ is the  complex--valued envelope of the wave, $\beta_1$, $\beta_2$ and $\beta_3$ are real parameters.
Decompose Eq. \eqref{sch2} into real and imaginary parts, \emph{i.e.},
\begin{equation}
\begin{aligned}
\Delta^1&\equiv u_{,t}+\frac{1}{2}v_{,xx}+v(u^2+v^2)\\
&+\varepsilon\left(\beta_1u_{,xxx}+\beta_2(u^2+v^2)u_{,x}+2\beta_3 u(uu_{,x}+vv_{,x})\right)=0,\\
\Delta^2&\equiv v_{,t}-\frac{1}{2}u_{,xx}-u(u^2+v^2)\\
&+\varepsilon\left(\beta_1v_{,xxx}+\beta_2(u^2+v^2)v_{,x}+2\beta_3 v(uu_{,x}+vv_{,x})\right)=0,
\end{aligned}
\end{equation}
and expand $u(t,x;\varepsilon)$ and $v(t,x;\varepsilon)$ at ﬁrst order in $\varepsilon$. By applying the approximate direct procedure,
we determine approximate multipliers with the form \eqref{second-order-mult} along with the associated approximate fluxes:
\begin{itemize}
\item
\begin{equation}
\begin{aligned}
\Lambda^1_1&=v_{(0),xx}+2v_{(0)}(u_{(0)}^2+v_{(0)}^2)\\
&+\varepsilon\left(v_{(1),xx}+2(\beta_2+2\beta_3-6\beta_1)(u_{(0)}^2+v_{(0)}^2)u_{(0),x}\right.\\
&+\left.2(u_{(0)}^2v_{(1)}+2u_{(0)}u_{(1)}v_{(0)}+3v^2_{(0)}v_{(1)})\right),\\
\Lambda^2_1&=-u_{(0),xx}-2u_{(0)}(u_{(0)}^2+v_{(0)}^2)\\
&-\varepsilon\left(u_{(1),xx}-2(\beta_2+2\beta_3-6\beta_1)(u_{(0)}^2+v_{(0)}^2)v_{(0),x}\right.\\
&+\left.2(v^2_{(0)}u_{(1)}+2u_{(0)}v_{(0)}v_{(1)}+3u_{(0)}^2u_{(1)})\right),
\end{aligned}
\end{equation}
with
\begin{align}
\Phi^t_1&=\frac{1}{2}\left(u_{(0),x}^2+v_{(0),x}^2\right)+2xu_{(0)}^3u_{(0),x}-\frac{v_{(0)}^2}{2}(2u_{(0)}^2+v_{(0)}^2)\nonumber\allowdisplaybreaks\\
&+\varepsilon\left(\left(u_{(1),x}+6xu_{(0)}^2u_{(1)}\right)u_{(0),x}\phantom{\frac{1}{c^2}}\right.\nonumber\allowdisplaybreaks\\
&+\left.\left(v_{(1),x}+\frac{2}{3}(\beta_2+2\beta_3-6\beta_1)(u_{(0)}^2+3v_{(0)}^2)u_{(0)}\right)v_{(0),x}\right.\nonumber\allowdisplaybreaks\\
&\left.+\phantom{\frac{1}{c}}2xu_{(0)}^3u_{(1),x}-2v_{(0)}(u_{(0)}^2v_{(1)}+u_{(0)}v_{(0)}u_{(1)}+v_{(0)}^2v_{(1)})\right),\nonumber\allowdisplaybreaks\\
\Phi^x_1&=-(u_{(0),t}u_{(0),x}+v_{(0),t}v_{(0),x}+2xu_{(0)}^3u_{(0),t})\\
&-\varepsilon\left(\frac{\beta_1}{2}(u_{(0),xx}^2+v_{(0),xx}^2)+2\beta_1(u_{(0)}^2+v_{(0)}^2)(u_{(0)}u_{(0),xx}+v_{(0)}v_{(0),xx})\right.\nonumber\allowdisplaybreaks\\
&+(u_{(0),x}+2xu_{(0)}^3)u_{(1),t}+v_{(0),x}v_{(1),t}+\left(u_{(1),x}+6xu_{(0)}^2u_{(1)}\right)u_{(0),t}\nonumber\allowdisplaybreaks\\
&+\left(v_{(1),x}+\frac{2}{3}(\beta_2+2\beta_3-6\beta_1)(u_{(0)}^2+3v_{(0)}^2)u_{(0)}\right)v_{(0),t}\nonumber\allowdisplaybreaks\\
&\left.+\phantom{\frac{1}{c}}(2\beta_1-\beta_3)\left(v_{(0)}u_{(0),x}-u_{(0)}v_{(0),x}\right)^2+2\beta_1(u_{(0)}^2+v_{(0)}^2)^3\right)\nonumber\allowdisplaybreaks;
\end{align}
\item
\begin{equation}
\begin{aligned}
\Lambda^1_2&=tu_{(0),x}+xv_{(0)}\\
&+\varepsilon\left((3\beta_1-\beta_2)tv_{(0),xx}+(\beta_2-6\beta_1)xu_{(0),x}+tu_{(1),x}+xv_{(1)}\phantom{\frac{}{}}\right.\\
&+\left.2(3\beta_1-\beta_2-\beta_3)(u_{(0)}^2+v_{(0)}^2)tv_{(0)}+\frac{u_{(0)}}{2}(\beta_2-6\beta_1)\right),\\
\Lambda^2_2&=tv_{(0),x}-xu_{(0)}\\
&-\varepsilon\left((3\beta_1-\beta_2)tu_{(0),xx}-(\beta_2-6\beta_1)xv_{(0),x}-tv_{(1),x}+xu_{(1)}\phantom{\frac{}{}}\right.\\
&\left.+2(3\beta_1-\beta_2-\beta_3)(u_{(0)}^2+v_{(0)}^2)tu_{(0)}-\frac{v_{(0)}}{2}(\beta_2-6\beta_1)\right),
\end{aligned}
\end{equation}
with
\begin{align}
\Phi^t_2&=\left(\frac{x^2}{2}u_{(0)}-tv_{(0)}\right)u_{(0),x}-\frac{x}{2}v_{(0)}^2\nonumber\allowdisplaybreaks\\
&+\varepsilon\left(\frac{3\beta_1-\beta_2}{2}t(u_{(0),x}^2+v_{(0),x}^2)+\left(\frac{x^2}{2}u_{(1)}-tv_{(1)}\right)u_{(0),x}\right.\nonumber\allowdisplaybreaks\\
&+\left(2(3\beta_1-\beta_2-\beta_3)tv_{(0)}^3-(6\beta_1-\beta_2)u_{(0)}\right)xv_{(0),x}\nonumber\allowdisplaybreaks\\
&+\left(\frac{x^2}{2}u_{(0)}-tv_{(0)}\right)u_{(1),x}-xv_{(0)}v_{(1)}\nonumber\allowdisplaybreaks\\
&-\left.\frac{3\beta1-\beta_2-\beta_3}{2}t(u_{(0)}^2+2v_{(0)}^2)u_{(0)}^2+\frac{\beta_2-6\beta_1}{2}u_{(0)}v_{(0)}\right),\nonumber\allowdisplaybreaks\\
\Phi^x_2&=-\left(\frac{x^2}{2}u_{(0)}-tv_{(0)}\right)u_{(0),t}+\frac{t}{4}(u_{(0),x}^2+v_{(0),x}^2)\nonumber\\
&+\frac{x}{2}(v_{(0)}u_{(0),x}-u_{(0)}v_{(0),x})+\frac{t}{4}(u_{(0)}^2+v_{(0)}^2)^2+\frac{u_{(0)}v_{(0)}}{2}\allowdisplaybreaks\\
&-\varepsilon\left(\beta_1((tu_{(0),x}+xv_{(0)})v_{(0),xx}-(tv_{(0),x}-xu_{(0)})u_{(0),xx})\phantom{\frac{1}{2}}\right.\nonumber\allowdisplaybreaks\\
&+t(3\beta_1-\beta_2)(u_{(0),t}u_{(0),x}+v_{(0),t}v_{(0),x})-\frac{x}{4}(\beta_2-4\beta_1)(u_{(0),x}^2+v_{(0),x}^2)    \nonumber\allowdisplaybreaks\\
&+\left(\frac{x^2}{2}u_{(1)}-tv_{(1)}\right)u_{(0),t}+(2(3\beta_1-\beta_2-\beta_3)tv_{(0)}^3+(\beta_2-6\beta_1)u_{(0)})xv_{(0),t}
\nonumber\allowdisplaybreaks\\
&-\left(\frac{t}{2}(u_{(1),x}-2\beta_3(u_{(0)}^2+v_{(0)}^2)v_{(0)})+\frac{1}{4}((\beta_2-2\beta_1)u_{(0)}+2xv_{(1)})\right)u_{(0),x}\nonumber\allowdisplaybreaks\\
&-\left(\frac{t}{2}(v_{(1),x}+2\beta_3(u_{(0)}^2+v_{(0)}^2)u_{(0)})+\frac{1}{4}((\beta_2-2\beta_1)v_{(0)}-2xu_{(1)})\right)v_{(0),x}\nonumber\allowdisplaybreaks\\
&+\left(\frac{x^2}{2}u_{(0)}-tv_{(0)}\right)u_{(1),t}-\frac{x}{2}(v_{(0)}u_{(1),x}-u_{(0)}v_{(1),x})\nonumber\allowdisplaybreaks\\
&-t(u_{(0)}^2+v_{(0)}^2)(u_{(0)}u_{(1)}+v_{(0)}v_{(1)})-\frac{1}{2}(u_{(0)}v_{(1)}+v_{(0)}u_{(1)})\nonumber\\
&+\left.\frac{x}{2}\left(3\beta_1((u_{(0)}^2+v_{(0)}^2)^2+v_{(0)}^4)-\beta_2 v_{(0)}^4+\beta_3 u_{(0)}^2(u_{(0)}^2+2v_{(0)}^2)\right)\right);\nonumber
\end{align}
\item
\begin{equation}
\begin{aligned}
\Lambda^1_3&=v_{(0)}+\varepsilon\left(v_{(1)}-3\beta_1u_{(0),x}\right),\\
\Lambda^2_3&=-u_{(0)}-\varepsilon\left(u_{(1)}+3\beta_1v_{(0),x}\right),
\end{aligned}
\end{equation}
with
\begin{equation}
\begin{aligned}
\Phi^t_3&=xu_{(0)}u_{(0),x}-\frac{v_{(0)}^2}2\\
&+\varepsilon\left(u_{(0),x}u_{(1),x}-3\beta_1u_{(0)}v_{(0),x}+xu_{(0)}u_{(1),x}-v_{(0)}v_{(1)}\right),\\
\Phi^x_3&=-\frac{u_{(0)}}{2}(2xu_{(0),t}+v_{(0),x})+\frac{v_{(0)}}{2}u_{(0),x}\\
&-\varepsilon\left(\beta_1(u_{(0)}u_{(0),xx}+v_{(0)}v_{(0),xx})+xu_{(1)}u_{(0),t}-3\beta_1u_{(0)}v_{(0),t}\phantom{\frac{1}{2}}\right.\\
&+\frac{\beta_1}{4}(u_{(0),x}^2+v_{(0),x}^2)-\frac{1}{2}(v_{(1)}u_{(0),x}-u_{(1)}v_{(0),x})+xu_{(0)}u_{(1),t}\\
&-\left.\frac{1}{2}(v_{(0)}u_{(1),x}-u_{(0)}v_{(1),x})+\frac{3\beta_1+\beta_2+2\beta_3}{4}(u_{(0)}^2+v_{(0)}^2)^2\right);
\end{aligned}
\end{equation}
\item
\begin{equation}
\begin{aligned}
\Lambda^1_4&=u_{(0),x}+\varepsilon(u_{(1),x}-2\beta_3v_{(0)}(u_{(0)}^2+v_{(0)}^2)),\\
\Lambda^2_4&=v_{(0),x}+\varepsilon(v_{(1),x}+2\beta_3u_{(0)}(u_{(0)}^2+v_{(0)}^2)),
\end{aligned}
\end{equation}
with
\begin{align}
\Phi^t_4&=v_{(0)}u_{(0),x}\nonumber\allowdisplaybreaks\\
&+\varepsilon\left(v_{(1)}u_{(0),x}+2\beta_3xv_{(0)}^3v_{(0),x}+v_{(0)}u_{(1),x}-\frac{\beta_3}{2}u_{(0)}^2(u_{(0)}^2+2v_{(0)}^2)\right),\nonumber\allowdisplaybreaks\\
\Phi^x_4&=-v_{(0)}u_{(0),t}-\frac{1}{4}(u_{(0),x}^2+v_{(0),x}^2+(u_{(0)}^2+v_{(0)}^2)^2)\allowdisplaybreaks\\
&-\varepsilon\left(\beta_1(v_{(0),x}u_{(0),xx}-u_{(0),x}v_{(0),xx})+v_{(1)}u_{(0),t}+2\beta_3xv_{(0)}^3v_{(0),t}\phantom{\frac{}{}}\right.\nonumber\allowdisplaybreaks\\
&+\frac{1}{2}(u_{(1),x}-2\beta_3v_{(0)}(u_{(0)}^2+v_{(0)}^2))u_{(0),x}\nonumber\allowdisplaybreaks\\
&+\frac{1}{2}(v_{(1),x}+2\beta_3u_{(0)}(u_{(0)}^2+v_{(0)}^2))v_{(0),x}\nonumber\allowdisplaybreaks\\
&\left.+\phantom{\frac{}{}}v_{(0)}u_{(1),t}+(u_{(0)}^2+v_{(0)}^2)(u_{(0)}u_{(1)}+v_{(0)}v_{(1)})\right).\nonumber
\end{align}
\end{itemize}
The remaining approximate multipliers do not produce new independent approximate conservation laws; in fact, we obtain:
\begin{equation}
\Lambda^i_{j+4}=\varepsilon \Lambda^i_j,
\end{equation}
with
\begin{equation}
\Phi^t_{j+4}=\varepsilon\Phi^t_j,\qquad \Phi^x_{j+4}=\varepsilon\Phi^x_j,
\end{equation}
where $i=1,2$ and $j=1,\dots,4$.
\begin{remark}
We also checked the existence of approximate multipliers depending on third order derivatives, but they turn out to be trivial, and we omit to report them.
\end{remark}

\subsection{Generalized Kaup--Newell equation}
Let us consider the generalized Kaup--Newell equation
\begin{equation}\label{kaup}
\begin{aligned}
\Delta^1&\equiv u_{,t}-\frac{1}{2}u_{,xx}+uvu_{,x}+\frac{1}{2}u^2v_{,x}+2\varepsilon uu_{,x}=0,\\
\Delta^2&\equiv v_{,t}+\frac{1}{2}v_{,xx}+uvv_{,x}+\frac{1}{2}v^2u_{,x}+2\varepsilon(vu_{,x}+uv_{,x})=0,
\end{aligned}
\end{equation}
where $u\equiv u(t,x;\varepsilon)$ and $v\equiv v(t,x;\varepsilon)$. Conservation laws for this system have been determined in \cite{Young2016} where a Lax pair has been used, and in \cite{Abdullahi2018} by means of a multiplier approach. 

By expanding the dependent variables at ﬁrst order in $\varepsilon$, and applying the approximate direct method,
we determine approximate multipliers with the form \eqref{second-order-mult} and the associated approximate fluxes:
\begin{itemize}
\item
\begin{align}
\Lambda^1_1&=2tv_{(0),xx}+2(3tu_{(0)}v_{(0)}-x)v_{(0),x}+(3tu_{(0)}^2v_{(0)}^2-2xu_{(0)}v_{(0)}-2)v_{(0)}\nonumber\allowdisplaybreaks\\
&+\varepsilon\left(2tv_{(1),xx}+6t(u_{(0)}v_{(1)}+v_{(0)}u_{(1)}+2u_{(0)})v_{(0),x}\phantom{\frac{}{}}\right.\nonumber\allowdisplaybreaks\\
&+2(3tu_{(0)}v_{(0)}-x)v_{(1),x}+3tu_{(0)}v_{(0)}^2(3u_{(0)}v_{(1)}+2v_{(0)}u_{(1)}+6u_{(0)})\nonumber\allowdisplaybreaks\\
&\left.-\phantom{\frac{}{}}2xv_{(0)}(2u_{(0)}v_{(1)}+v_{(0)}u_{(1)}+4u_{(0)})-2v_{(1)}\right),\allowdisplaybreaks\\
\Lambda^2_1&=2tu_{(0),xx}-2(3tu_{(0)}v_{(0)}-x)u_{(0),x}+(3tu_{(0)}v_{(0)}-2x)u_{(0)}^2v_{(0)}\nonumber\allowdisplaybreaks\\
&+\varepsilon\left(2tu_{(1),xx}-6t(u_{(0)}v_{(1)}+v_{(0)}u_{(1)}+2u_{(0)})u_{(0),x}\phantom{\frac{}{}}\right.\nonumber\allowdisplaybreaks\\
&-2(3tu_{(0)}v_{(0)}-x)u_{(1),x}+3tu_{(0)}^2v_{(0)}(2u_{(0)}v_{(1)}+3v_{(0)}u_{(1)}+4u_{(0)})\nonumber\allowdisplaybreaks\\
&\left.-\phantom{\frac{}{}}2xu_{(0)}(u_{(0)}v_{(1)}+2v_{(0)}u_{(1)}+2u_{(0)})\right),\nonumber
\end{align}
with
\begin{align}
\Phi^t_1&=2tu_{(0),x}v_{(0),x}-\frac{(3tu_{(0)}v_{(0)}-2x)^2}{4}u_{(0)}^2v_{(0)}v_{(0),x}\nonumber\\
&-\frac{(3tu_{(0)}^2v_{(0)}^2-2xu_{(0)}v_{(0)}-4)(3tu_{(0)}v_{(0)}-2x)}{4}v_{(0)}u_{(0),x}\nonumber\allowdisplaybreaks\\
&+\varepsilon\left(2t(u_{(0),x}v_{(1),x}+v_{(0),x}u_{(1),x})\phantom{\frac{1}{1}}\right.\nonumber\allowdisplaybreaks\\
&-\left(\frac{9}{4}(4u_{(0)}v_{(1)}+3v_{(0)}u_{(1)}+8u_{(0)})t^2u_{(0)}^2v_{(0)}^3\right.\nonumber\allowdisplaybreaks\\
&-3(3u_{(0)}v_{(1)}+2v_{(0)}u_{(1)}+6u_{(0)})txu_{(0)}v_{(0)}^2\nonumber\allowdisplaybreaks\\
&+(2u_{(0)}v_{(1)}+v_{(0)}u_{(1)}+4u_{(0)})x^2v_{(0)}\nonumber\allowdisplaybreaks\\
&\left.-\phantom{\frac{1}{1}}(6u_{(0)}v_{(1)}+3v_{(0)}u_{(1)}+8u_{(0)})tv_{(0)}+2xv_{(1)}\right)u_{(0),x}\nonumber\\
&-\left(\frac{9}{4}(3u_{(0)}v_{(1)}+4v_{(0)}u_{(1)}+6u_{(0)})t^2u_{(0)}^3v_{(0)}^2\right.\nonumber\allowdisplaybreaks\\
&-.3(2u_{(0)}v_{(1)}+3v_{(0)}u_{(1)}+4u_{(0)})txu_{(0)}^2v_{(0)}
\nonumber\\
&\left.+\phantom{\frac{1}{1}}(u_{(0)}v_{(1)}+2v_{(0)}u_{(1)}+2u_{(0)})x^2u_{(0)}+2tu_{(0)}^2\right)v_{(0),x}\nonumber\allowdisplaybreaks\\
&-\frac{(3tu_{(0)}^2v_{(0)}^2-2xu_{(0)}v_{(0)}-4)(3tu_{(0)}v_{(0)}-2x)}{4}v_{(0)}u_{(1),x}\nonumber\\
&-\left.\frac{(3tu_{(0)}v_{(0)}-2x)^2}{4}u_{(0)}^2v_{(0)}v_{(1),x}\right),\allowdisplaybreaks\\
\Phi^x_1&=-2t(u_{(0),t}v_{(0),x}+v_{(0),t}u_{(0),x})\nonumber\allowdisplaybreaks\\
&+\frac{(3tu_{(0)}^2v_{(0)}^2-2xu_{(0)}v_{(0)}-4)(3tu_{(0)}v_{(0)}-2x)}{4}v_{(0)}u_{(0),t}\nonumber\allowdisplaybreaks\\
&+\frac{(3tu_{(0)}v_{(0)}-2x)^2}{4}u_{(0)}^2v_{(0)}v_{(0),t}-\frac{t}{2}(v_{(0)}^2u_{(0),x}^2+u_{(0)}^2v_{(0),x}^2)\nonumber\allowdisplaybreaks\\
&+(tu_{(0)}v_{(0)}-x)u_{(0),x}v_{(0),x}-\frac{3tu_{(0)}v_{(0)}-2x}{2}u_{(0)}^2v_{(0)}v_{(0),x}\nonumber\allowdisplaybreaks\\
&+\frac{3tu_{(0)}^2v_{(0)}^2-2xu_{(0)}v_{(0)}-2}{2}v_{(0)}u_{(0),x}\nonumber\allowdisplaybreaks\\
&-\varepsilon\left(2t(u_{(0),t}v_{(1),x}+v_{(0),t}u_{(1),x})\phantom{\frac{1}{2}}\right.\nonumber\allowdisplaybreaks\\
&+\frac{9}{4}(4u_{(0)}v_{(1)}+3v_{(0)}u_{(1)}+8u_{(0)})t^2u_{(0)}^2v_{(0)}^3u_{(0),t}\nonumber\allowdisplaybreaks\\
&+\frac{9}{4}(3u_{(0)}v_{(1)}+4v_{(0)}u_{(1)}+6u_{(0)})t^2u_{(0)}^3v_{(0)}^2v_{(0),t}-t((v_{(1)}+2)v_{(0)}u_{(0),x}^2\nonumber\allowdisplaybreaks\\
&-(u_{(0)}v_{(1)}+v_{(0)}u_{(1)}+2u_{(0)})u_{(0),x}v_{(0),x}+u_{(0)}u_{(1)}v_{(0),x}^2)\nonumber\allowdisplaybreaks\\
&-2t(u_{(0),x}v_{(1),t}+v_{(0),x}u_{(1),t})-t(v_{(0)}^2u_{(0),x}u_{(1),x}+u_{(0)}^2v_{(0),x}v_{(1),x})\nonumber\allowdisplaybreaks\\
&+(tu_{(0)}v_{(0)}-x)(u_{(0),x}v_{(1),x}+v_{(0),x}u_{(1),x})\nonumber\allowdisplaybreaks\\
&+\left(\frac{3}{2}(3u_{(0)}v_{(1)}+2v_{(0)}u_{(1)}+6u_{(0)})tu_{(0)}v_{(0)}^2\right.\nonumber\allowdisplaybreaks\\
&\left.-\phantom{\frac{}{}}(2u_{(0)}v_{(1)}+v_{(0)}u_{(1)}+4u_{(0)})xu_{(0)}-v_{(1)}\right)u_{(0),x}\nonumber\allowdisplaybreaks\\
&-\left(\frac{3}{2}(2u_{(0)}v_{(1)}+3v_{(0)}u_{(1)}+4u_{(0)})tu_{(0)}^2v_{(0)}\right.\nonumber\allowdisplaybreaks\\
&\left.-\phantom{\frac{}{}}(u_{(0)}v_{(1)}+2v_{(0)}u_{(1)}+2u_{(0)})xu_{(0)}\right)v_{(0),x}\nonumber\allowdisplaybreaks\\
&+\frac{(3tu_{(0)}^2v_{(0)}^2-2xu_{(0)}v_{(0)}-4)(3tu_{(0)}v_{(0)}-2x)}{4}v_{(0)}u_{(1),t}\nonumber\allowdisplaybreaks\\
&+\frac{3tu_{(0)}^2v_{(0)}^2-2xu_{(0)}v_{(0)}-2}{2}v_{(0)}u_{(1),x}\nonumber\allowdisplaybreaks\\
&+\left.\frac{(3tu_{(0)}v_{(0)}-2x)^2}{4}u_{(0)}^2v_{(0)}v_{(1),t}-\frac{3tu_{(0)}v_{(0)}-2x}{2}u_{(0)}^2v_{(0)}v_{(1),x}\right);\nonumber\allowdisplaybreaks
\end{align}

\item
\begin{equation}
\begin{aligned}
\Lambda^1_2&=2v_{(0),xx}+6u_{(0)}v_{(0)}v_{(0),x}+3u_{(0)}^2v_{(0)}^3\\
&+\varepsilon\left(2v_{(1),xx}+6(u_{(0)}v_{(1)}+v_{(0)}u_{(1)}+2u_{(0)})v_{(0),x}\phantom{\frac{}{}}\right.\\
&+\left.6u_{(0)}v_{(0)}v_{(1),x}\phantom{\frac{}{}}+3(3u_{(0)}v_{(1)}+2v_{(0)}u_{(1)}+6u_{(0)})u_{(0)}v_{(0)}^2\right),\\
\Lambda^2_2&=2u_{(0),xx}-6u_{(0)}v_{(0)}u_{(0),x}+3u_{(0)}^3v_{(0)}^2\\
&+\varepsilon\left(2u_{(1),xx}-6(u_{(0)}v_{(1)}+v_{(0)}u_{(1)}+2u_{(0)})u_{(0),x}\phantom{\frac{}{}}\right.\\
&-\left.6u_{(0)}v_{(0)}u_{(1),x}\phantom{\frac{}{}}+3(2u_{(0)}v_{(1)}+3v_{(0)}u_{(1)}+4u_{(0)})u_{(0)}^2v_{(0)}\right),
\end{aligned}
\end{equation}
with
\begin{align}
\Phi^t_2&=2u_{(0),x}v_{(0),x}-\frac{3}{2}(3tu_{(0)}v_{(0)}-2x)u_{(0)}^2v_{(0)}^3u_{(0),x}\nonumber\allowdisplaybreaks\\
&-\frac{3}{2}(3tu_{(0)}^2v_{(0)}^2-2xu_{(0)}v_{(0)}+2)u_{(0)}^2v_{(0)}v_{(0),x}\nonumber\allowdisplaybreaks\\
&+\varepsilon\left(2(u_{(0),x}v_{(1),x}+v_{(0),x}u_{(1),x})\phantom{\frac{1}{2}}\right.\nonumber\allowdisplaybreaks\\
&-\left(\frac{9}{2}(4u_{(0)}v_{(1)}+3v_{(0)}u_{(1)}+8u_{(0)})tu_{(0)}^2v_{(0)}^3\right.\nonumber\allowdisplaybreaks\\
&-\left.\phantom{\frac{}{}}3(3u_{(0)}v_{(1)}+2v_{(0)}u_{(1)}+6u_{(0)})xu_{(0)}v_{(0)}^2-12u_{(0)}v_{(0)}\right)u_{(0),x}\nonumber\allowdisplaybreaks\\
&-\left(\frac{9}{2}(3u_{(0)}v_{(1)}+4v_{(0)}u_{(1)}+6u_{(0)})tu_{(0)}^3v_{(0)}^2\right.\nonumber\allowdisplaybreaks\\
&-3(2u_{(0)}v_{(1)}+3v_{(0)}u_{(1)}+4u_{(0)})xu_{(0)}^2v_{(0)}\nonumber\allowdisplaybreaks\\
&\left.+\phantom{\frac{}{}}3(u_{(0)}v_{(1)}+2v_{(0)}u_{(1)})u_{(0)}\right)v_{(0),x}-\frac{3}{2}(3tu_{(0)}v_{(0)}-2x)u_{(0)}^2v_{(0)}^3u_{(1),x}\nonumber\allowdisplaybreaks\\
&-\left.\frac{3}{2}(3tu_{(0)}^2v_{(0)}^2-2xu_{(0)}v_{(0)}-2)u_{(0)}^2v_{(0)}v_{(1),x}\right),\allowdisplaybreaks\\
\Phi^x_2&=-2(u_{(0),t}v_{(0),x}+v_{(0),t}u_{(0),x})+\frac{3}{2}(3tu_{(0)}v_{(0)}-2x)u_{(0)}^2v_{(0)}^3u_{(0),t}\nonumber\allowdisplaybreaks\\
&+\frac{3}{2}(3tu_{(0)}^2v_{(0)}^2-2xu_{(0)}v_{(0)}+2)u_{(0)}^2v_{(0)}v_{(0),t}\nonumber\allowdisplaybreaks\\
&-\frac{1}{2}(v_{(0)}u_{(0),x}-u_{(0)}v_{(0),x})^2+\frac{3}{2}u_{(0)}^2v_{(0)}^2(v_{(0)}u_{(0),x}-u_{(0)}v_{(0),x})\nonumber\allowdisplaybreaks\\
&+\varepsilon\left(-2(u_{(0),t}v_{(1),x}+v_{(0),t}u_{(1),x})\phantom{\frac{1}{2}}\right.\nonumber\allowdisplaybreaks\\
&+\left(\frac{9}{2}(4u_{(0)}v_{(1)}+3v_{(0)}u_{(1)}+8u_{(0)})tu_{(0)}^2v_{(0)}^3\right.\nonumber\allowdisplaybreaks\\
&-\left.\phantom{\frac{}{}}3(3u_{(0)}v_{(1)}+2v_{(0)}u_{(1)}+6u_{(0)})xu_{(0)}v_{(0)}^2-12u_{(0)}v_{(0)}\right)u_{(0),t}\nonumber\allowdisplaybreaks\\
&+\left(\frac{9}{2}(3u_{(0)}v_{(1)}+4v_{(0)}u_{(1)}+6u_{(0)})tu_{(0)}^3v_{(0)}^2\right.\nonumber\allowdisplaybreaks\\
&-3(2u_{(0)}v_{(1)}+3v_{(0)}u_{(1)}+4u_{(0)})xu_{(0)}^2v_{(0)}\nonumber\allowdisplaybreaks\\
&+\left.\phantom{\frac{}{}}3(u_{(0)}v_{(1)}+2v_{(0)}u_{(1)})u_{(0)}\right)v_{(0),t}
-2(u_{(0),x}v_{(1),t}+v_{(0),x}u_{(1),t})\nonumber\allowdisplaybreaks\\
&-(v_{(0)}u_{(0),x}-u_{(0)}v_{(0),x})((v_{(1)}+2)u_{(0),x}-u_{(1)}v_{(0),x})\nonumber\allowdisplaybreaks\\
&+(v_{(0)}u_{(0),x}-u_{(0)}v_{(0),x})(u_{(0)}v_{(1),x}-v_{(0)}u_{(1),x})\nonumber\allowdisplaybreaks\\
&+\frac{3}{2}\left((3u_{(0)}v_{(1)}+2v_{(0)}u_{(1)}+6u_{(0)})u_{(0)}v^2_{(0)}u_{(0),x}+\phantom{\frac{1}{1}}\right.\nonumber\allowdisplaybreaks\\
&-(2u_{(0)}v_{(1)}+3v_{(0)}u_{(1)}+4u_{(0)})u_{(0)}^2v_{(0)}v_{(0),x}\nonumber\allowdisplaybreaks\\
&+(3tu_{(0)}v_{(0)}-2x)u_{(0)}^2v_{(0)}^3u_{(1),t}+u_{(0)}^2v_{(0)}^2(v_{(0)}u_{(1),x}-u_{(0)}v_{(1),x})\nonumber\allowdisplaybreaks\\
&\left.\left.+\phantom{\frac{1}{2}}(3tu_{(0)}^2v_{(0)}^2-2xu_{(0)}v_{(0)}+2)u_{(0)}^2v_{(0)}v_{(1),t}\right)\right);\nonumber\allowdisplaybreaks
\end{align}

\item
\begin{equation}
\begin{aligned}
\Lambda^1_3&=v_{(0),x}+u_{(0)}v_{(0)}^2\\
&+\varepsilon\left(v_{(1),x}+(u_{(1)}v_{(0)}+2u_{(0)}v_{(1)}+4u_{(0)})v_{(0)}\right),\\
\Lambda^2_3&=-u_{(0),x}+u_{(0)}^2v_{(0)}\\
&-\varepsilon\left(u_{(1),x}-\left(2u_{(1)}v_{(0)}+2u_{(0)}+u_{(0)}v_{(1)}\right)u_{(0)}\right),
\end{aligned}
\end{equation}
with
\begin{align}
\Phi^t_3&=\frac{1}{2}\left((3tu_{(0)}v_{(0)}-2x)u_{(0)}v_{(0)}^2u_{(0),x}\right.\nonumber\allowdisplaybreaks\\
&+\left.(3tu_{(0)}^2v_{(0)}^2-2xu_{(0)}v_{(0)}+2)u_{(0)}v_{(0),x}\right)\nonumber\allowdisplaybreaks\\
&+\varepsilon\left(\left(\frac{3}{2}(3u_{(0)}v_{(1)}+2u_{(1)}v_{(0)}+6u_{(0)})tu_{(0)}v_{(0)}^2\right.\right.\nonumber\allowdisplaybreaks\\
&\left.-\phantom{\frac{1}{2}}(2u_{(0)}v_{(1)}+u_{(1)}v_{(0)}+4u_{(0)})xv_{(0)}\right)u_{(0),x}\nonumber\allowdisplaybreaks\\
&+\left(\frac{3}{2}(2u_{(0)}v_{(1)}+3u_{(1)}v_{(0)}+4u_{(0)})tu_{(0)}^2v_{(0)}\right.\nonumber\allowdisplaybreaks\\
&\left.-\phantom{\frac{1}{2}}(u_{(0)}v_{(1)}+2u_{(1)}v_{(0)}+2u_{(0)})xu_{(0)}+u_{(1)}\right)v_{(0),x}\nonumber\allowdisplaybreaks\\
&+\frac{1}{2}(3tu_{(0)}v_{(0)}-2x)u_{(0)}v_{(0)}^2u_{(1),x}\nonumber\allowdisplaybreaks\\
&+\left.\frac{1}{2}(3tu_{(0)}^2v_{(0)}^2-2xu_{(0)}v_{(0)}+2)u_{(0)}v_{(1),x}\right),\allowdisplaybreaks\\
\Phi^x_3&=\frac{1}{2}\left((2x-3tu_{(0)}v_{(0)})u_{(0)}v_{(0)}^2u_{(0),t}-(v_{(0),x}+u_{(0)}v_{(0)}^2)u_{(0),x}\right.\nonumber\allowdisplaybreaks\\
&-\left.(3tu_{(0)}^2v_{(0)}^2-2xu_{(0)}v_{(0)}+2)u_{(0)}v_{(0),t}+u_{(0)}^2v_{(0)}v_{(0),x}\right)\nonumber\allowdisplaybreaks\\
&-\varepsilon\left(\left(\frac{3}{2}(3u_{(0)}v_{(1)}+2u_{(1)}v_{(0)}+6u_{(0)})tu_{(0)}v_{(0)}^2\right.\right.\nonumber\allowdisplaybreaks\\
&\left.-\phantom{\frac{1}{2}}(2u_{(0)}v_{(1)}+u_{(1)}v_{(0)}+4u_{(0)})xv_{(0)}\right)u_{(0),t}\nonumber\allowdisplaybreaks\\
&+\left(\frac{3}{2}(2u_{(0)}v_{(1)}+3u_{(1)}v_{(0)}+4u_{(0)})tu_{(0)}^2v_{(0)}\right.\nonumber\allowdisplaybreaks\\
&\left.-\phantom{\frac{1}{2}}(u_{(0)}v_{(1)}+2u_{(1)}v_{(0)}+2u_{(0)})xu_{(0)}+u_{(1)}\right)v_{(0),t}\nonumber\allowdisplaybreaks\\
&+\frac{1}{2}(v_{(1),x}+(2u_{(0)}v_{(1)}+u_{(1)}v_{(0)}+4u_{(0)})v_{(0)})u_{(0),x}\nonumber\allowdisplaybreaks\\
&-\frac{1}{2}(u_{(0)}v_{(1)}+2u_{(1)}v_{(0)}+2u_{(0)})u_{(0)}v_{(0),x}\nonumber\allowdisplaybreaks\\
&+\frac{1}{2}(3tu_{(0)}v_{(0)}-2x)u_{(0)}v_{(0)}^2u_{(1),t}+\frac{1}{2}(v_{(0),x}+u_{(0)}v_{(0)}^2)u_{(1),x}\nonumber\allowdisplaybreaks\\
&+\left.\frac{1}{2}(3tu_{(0)}^2v_{(0)}^2-2xu_{(0)}v_{(0)}+2)u_{(0)}v_{(1),t}-\frac{1}{2}u_{(0)}^2v_{(0)}v_{(1),x}\right);\nonumber\allowdisplaybreaks
\end{align}

\item
\begin{equation}
\begin{aligned}
\Lambda^1_4&=v_{(0)}+\varepsilon v_{(1)},\qquad\Lambda^2_4=u_{(0)}+\varepsilon u_{(1)},
\end{aligned}
\end{equation}
with
\begin{align}
\Phi^t_4&=\frac{1}{2}(3tu_{(0)}v_{(0)}-2x)(v_{(0)}u_{(0),x}+u_{(0)}v_{(0),x})\nonumber\allowdisplaybreaks\\
&+\frac{\varepsilon}{2}\left(((6u_{(0)}v_{(1)}+3v_{(0)}u_{(1)}+8u_{(0)})tv_{(0)}-2xv_{(1)})u_{(0),x}\phantom{\frac{}{}}\right.\nonumber\allowdisplaybreaks\\
&+((3u_{(0)}v_{(1)}+6v_{(0)}u_{(1)}+4u_{(0)})tu_{(0)}-2xu_{(1)})v_{(0),x}\nonumber\allowdisplaybreaks\\
&\left.+\phantom{\frac{}{}}(3tu_{(0)}v_{(0)}-2x)(v_{(0)}u_{(1),x}+u_{(0)}v_{(1),x})\right),\\
\Phi^x_4&=-\frac{1}{2}\left((3tu_{(0)}v_{(0)}-2x)(v_{(0)}u_{(0),t}+u_{(0)}v_{(0),t})+v_{(0)}u_{(0),x}-u_{(0)}v_{(0),x}\right)\nonumber\allowdisplaybreaks\\
&-\frac{\varepsilon}{2}\left(((6u_{(0)}v_{(1)}+3v_{(0)}u_{(1)}+8u_{(0)})tv_{(0)}-2xv_{(1)})u_{(0),t}\phantom{\frac{}{}}\right.\nonumber\allowdisplaybreaks\\
&+((3u_{(0)}v_{(1)}+6v_{(0)}u_{(1)}+4u_{(0)})tu_{(0)}-2xu_{(1)})v_{(0),t}\nonumber\allowdisplaybreaks\\
&+(3tu_{(0)}v_{(0)}-2x)(v_{(0)}u_{(1),t}+u_{(0)}v_{(1),t})\nonumber\allowdisplaybreaks\\
&\left.+\phantom{\frac{}{}}v_{(0)}u_{(1),x}-u_{(0)}v_{(1),x}+v_{(1)}u_{(0),x}-u_{(1)}v_{(0),x}\right);\nonumber
\end{align}
\item
\begin{equation}
\Lambda^1_5=1,\qquad \Lambda^2_5=0,
\end{equation}
with
\begin{align}
\Phi^t_5&=(tu_{(0)}v_{(0)}-x)u_{(0),x}+\frac{1}{2}tu_{(0)}^2v_{(0),x}\nonumber\allowdisplaybreaks\\
&+\varepsilon\left(t(u_{(0)}v_{(1)}+u_{(1)}v_{(0)}+2u_{(0)})u_{(0),x}+tu_{(0)}u_{(1)}v_{(0),x}\phantom{\frac{1}{2}}\right.\nonumber\allowdisplaybreaks\\
&+\left.(tu_{(0)}v_{(0)}-x)u_{(1),x}+\frac{1}{2}tu_{(0)}^2v_{(1),x}\right),\\
\Phi^x_5&=-(tu_{(0)}v_{(0)}-x)u_{(0),t}-\frac{1}{2}tu_{(0)}^2v_{(0),t}-\frac{1}{2}u_{(0),x}\nonumber\allowdisplaybreaks\\
&-\varepsilon\left(t(u_{(0)}v_{(1)}+u_{(1)}v_{(0)}+2u_{(0)})u_{(0),t}+tu_{(0)}u_{(1)}v_{(0),t}\phantom{\frac{1}{2}}\right.\nonumber\allowdisplaybreaks\\
&+\left.(tu_{(0)}v_{(0)}-x)u_{(1),t}+\frac{1}{2}tu_{(0)}^2v_{(1),t}+\frac{1}{2}u_{(1),x}\right);\nonumber
\end{align}

\item
\begin{equation}
\Lambda^1_6=0,\qquad\Lambda^2_6=1,
\end{equation}
with
\begin{align}
\Phi^t_6&=\frac{1}{2}tv_{(0)}^2u_{(0),x}+(tu_{(0)}v_{(0)}-x)v_{(0),x}\nonumber\allowdisplaybreaks\\
&+\varepsilon\left(t(v_{(1)}+2)v_{(0)}u_{(0),x}+t(u_{(0)}v_{(1)}+u_{(1)}v_{(0)}+2u_{(0)})v_{(0),x}\phantom{\frac{1}{2}}\right.\nonumber\allowdisplaybreaks\\
&+\left.\frac{1}{2}tv_{(0)}^2u_{(1),x}+(tu_{(0)}v_{(0)}-x)v_{(1),x}\right),\\
\Phi^x_6&=-\frac{1}{2}tv_{(0)}^2u_{(0),t}-(tu_{(0)}v_{(0)}-x)v_{(0),t}+\frac{1}{2}v_{(0),x}\nonumber\allowdisplaybreaks\\
&-\varepsilon\left(t(v_{(1)}+2)v_{(0)}u_{(0),t}+t(u_{(0)}v_{(1)}+u_{(1)}v_{(0)}+2u_{(0)})v_{(0),t}\phantom{\frac{1}{2}}\right.\nonumber\allowdisplaybreaks\\
&+\left.\frac{1}{2}tv_{(0)}^2u_{(1),t}+(tu_{(0)}v_{(0)}-x)v_{(1),t}-\frac{1}{2}{v_{(1),x}}\right).\nonumber
\end{align}
\end{itemize}
The remaining approximate multipliers do not produce new independent approximate conservation laws; in fact, we obtain:
\begin{equation}
\Lambda^i_{j+6}=\varepsilon \Lambda^i_j,
\end{equation}
with
\begin{equation}
\Phi^t_{j+6}=\varepsilon\Phi^t_j,\qquad \Phi^x_{j+6}=\varepsilon\Phi^x_j,
\end{equation}
where $i=1,2$ and $j=1,\dots,6$.
\section{Conclusions}\label{sec:6}
In this paper, we considered some perturbed differential equations not arising from a variational principle; by adapting some ideas recently proposed for dealing with approximate Lie symmetries of differential equations containing terms with a small order of magnitude \cite{DSGO-2018}, we proposed to extend to the approximate framework the direct method \cite{AncoBluman2002a,AncoBluman2002b} for obtaining conservation laws. 
Consistently with the approach \cite{DSGO-2018}, approximate multipliers, {assumed to be dependent on the small parameter}, have been defined, and the new approximate direct procedure has been exploited. 

A formulation of approximate multipliers, though within a different approach, has been given in \cite{Jamal-2019}. Therein, the dependent variables are not expanded in power series, and the expansion of Lagrange multipliers is not equal to the one here considered.

Here, after defining the approximate multipliers accordingly with the approach proposed in \cite{DSGO-2018}, a theorem ensuring the conditions allowing to recover the approximate multipliers has been proved, and, as a by--product, an approximate Euler operator has been used. 
We want to remark that this formulation of approximate multiplier and approximate Euler operator allow us to recover sets of approximate conservation laws that are coherent with perturbation analysis. In order to emphasize the methodological differences between the methods for approximate conservation laws already introduced in \cite{Jamal-2019}, which are in the framework of Baikov--Gazizov--Ibragimov \cite{BGI-1988} and Fushchich--Shtelen \cite{FS-1989} approaches for approximate symmetries, with the procedure here proposed, an example exhibiting the results with the three different methods is analyzed.

We also applied the general framework to some perturbed problems of physical interest not possessing a variational formulation. By limiting to first order approximation, we considered the KdV--Burgers equation, a nonlinear wave equation, two nonlinear Schr\"{o}dinger equations, and a generalized Kaup--Newell equation. For each example, approximate multipliers and the corresponding approximate conserved quantities have been derived. All the needed computations have been done by means of the program ReLie \cite{Oliveri-relie} written in the Computer Algebra System Reduce \cite{Reduce}. 

Work is in progress in order to apply the method proposed in \cite{DSGO-2018} to the derivation of new approaches for the construction of approximate conservation laws, when dealing with differential equations involving small terms not arising from a variational principle. In particular, we plan to investigate the 
self--adjoint approach \cite{Ibragimov2007}, and the methods of partial Lagrangians \cite{Kara-Mahomed-2006} and quasi Lagrangians \cite{Rosenhaus-Shankar-2019} into the approximate context. These extensions and generalizations will be the object of forthcoming papers.

\section*{Acknoledgments}
This work was supported by ``Gruppo Nazionale per la Fisica Matematica'' (GNFM) of ``Istituto Nazionale di Alta Matematica''. M.G. acknowledges the support through the ``Progetto Giovani No. GNFM 2020''.

\end{document}